\documentclass[journal]{IEEEtran}

\usepackage{amssymb,amsthm,pict2e,enumerate}
\usepackage{array}
\usepackage{booktabs, soul}
\usepackage{color}
\usepackage{bm}
\usepackage{amsfonts}
\usepackage{algorithm}
\usepackage{algorithmic}
\usepackage{amsmath,dsfont}
\usepackage{cite}
\usepackage{url}
\usepackage{acronym}

\ifCLASSINFOpdf
\usepackage[pdftex]{graphicx}
\else
\usepackage[dvips]{graphicx}
\fi

\usepackage{epstopdf}

\ifCLASSOPTIONcompsoc
\usepackage[caption=false,font=normalsize,labelfont=sf,textfont=sf]{subfig}
\else
\usepackage[caption=false,font=footnotesize]{subfig}
\fi

\newtheorem{Thm}{Theorem}

\newtheorem{lemma}[Thm]{Lemma}
\newtheorem{corollary}[Thm]{Corollary}

\newcommand{\mySet}[1]{\mathcal{#1}}
\newcommand{\Gain}{\alpha_K}

\newcommand{\Revise}[1]{{#1}}
\newcommand{\myfiguresize}{3}

\begin{document}
	
	\acrodef{snr}[SNR]{signal-to-noise ratio}
	\acrodef{mimo}[MIMO]{multiple-input multiple-output}
	\acrodef{mse}[MSE]{mean-squared error}
	\acrodef{pdf}[PDF]{probability density function}
	\acrodef{cs}[CS]{compressed sensing}
	\acrodef{wmar}[WMAR]{wideband multi-carrier agile radar}
	\acrodef{caesar}[CAESAR]{multi-Carrier AgilE phaSed Array Radar}
	\acrodef{icassp}[ICASSP]{International Conference on Acoustics, Speech, and Signal Processing}
	\acrodef{roc}[ROC]{receiver operating characteristic}
	\acrodef{awgn}[AWGN]{additive white Gaussian noise}
	\acrodef{dfrc}[DFRC]{dual function radar-communications}
	\acrodef{fmcw}[FMCW]{frequency-modulated continuous-wave}
	\acrodef{hrr}[HRR]{high range resolution}
	\acrodef{cpi}[CPI]{coherent processing interval}
	\acrodef{far}[FAR]{frequency agile radar}	
	\acrodef{rmse}[RMSE]{root mean squared error}	
	\acrodef{crc}[CRC]{coarse range cell}
	\acrodef{fda}[FDA]{frequency diversity array} 
	\acrodef{fdma}[FDMA]{frequency division multiple access} 
	\acrodef{fdmamimo}[FDMA-MIMO]{\ac{fdma} \ac{mimo}}
	\acrodef{summer}[SUMMeR]{sub-Nyquist \ac{mimo} radar}
	\acrodef{eccm}[ECCM]{electronic counter-countermeasures}
	\acrodef{emc}[EMC]{electromagnetic compatibility}

	%
	\title{Multi-Carrier Agile Phased Array Radar}
	%
	%
	%
	
	\author{Tianyao Huang, Nir Shlezinger, Xingyu Xu, Dingyou Ma, Yimin Liu, and Yonina C. Eldar 
		
		\thanks{Parts of this work were presented in the 2018 IEEE International Workshop on Compressed Sensing applied to Radar, Multimodal Sensing, and Imaging (CoSeRa).
			This work received funding from the National Natural Science Foundation of China under Grants 61571260 and 61801258, from the European Union’s Horizon 2020 research and innovation program under grant No. 646804-ERC-COG-BNYQ, and from the Air Force Office of Scientific Research under grant No. FA9550-18-1-0208.
		}
		\thanks{T. Huang, Y. Liu, X. Xu,  and D. Ma are with the EE Department, Tsinghua University, Beijing, China (e-mail:			\{huangtianyao, yiminliu\}@tsinghua.edu.cn; \{xy-xu15, mdy16\}@mails.tsinghua.edu.cn).
		}
		\thanks{
			N. Shlezinger and Y. C. Eldar are with the Faculty of Math and CS, Weizmann Institute of Science, Rehovot,  Israel (e-mail: nirshlezinger1@gmail.com; yonina.eldar@weizmann.ac.il).
		}
		\vspace{-1.1cm}
	}

	\maketitle
	
	\begin{abstract}
		Modern radar systems are expected to operate reliably in congested environments. 
		{\Revise{A candidate technology for meeting these demands is frequency agile radar (FAR), which randomly changes its carrier frequencies. FAR is known to improve the electronic counter-countermeasures (ECCM) performance while facilitating operation in  congested setups.}} 
		To enhance the target recovery performance of FAR in complex electromagnetic environments, we propose two radar schemes extending FAR to multi-carrier waveforms. The first is  Wideband Multi-carrier Agile Radar (WMAR), which transmits/receives wideband waveforms simultaneously with every antenna. To \Revise{mitigate} the \Revise{demanding hardware requirements} associated with wideband waveforms used by WMAR, we next propose multi-Carrier AgilE phaSed Array Radar (CAESAR). CAESAR uses narrowband monotone waveforms, thus \Revise{facilitating ease of implementation of the system}, while introducing {\em spatial agility}. 
		We characterize the transmitted and received signals of the proposed schemes, and develop an algorithm for recovering the targets, \Revise{based on} concepts from compressed sensing to estimate the range-Doppler parameters of the targets. We \Revise{then} derive conditions which guarantee their accurate reconstruction. 
		Our numerical study demonstrates that both multi-carrier schemes improve performance compared to FAR while maintaining its practical benefits. We also demonstrate that the performance of CAESAR, which uses monotone waveforms, is within a small gap from the wideband radar. 
	\end{abstract}

	%
	\IEEEpeerreviewmaketitle

	\vspace{-0.4cm}
	\section{Introduction}
	\label{sec:intro}
	\vspace{-0.1cm}
	{Modern radars must be reliable, but at the same time compact, flexible, robust, and efficient in terms of cost and power usage \cite{Axelsson2007,Yang2013, Liu2017DFRC,Cohen2018a,Cohen2018c}.  A possible approach to meet these requirements is by exploiting {\em frequency agility} \cite{Axelsson2007}, namely, to utilize narrowband waveforms, 
	while allowing the carrier frequencies to vary between different radar pulses. 
	\Revise{Among the main advantages of \ac{far} are its} excellent \ac{eccm} and \ac{emc} performance \cite{Axelsson2007}, and the fact that it has the flexibility of supporting spectrum sharing \cite{Cohen2018a}.} Furthermore, \ac{far} is compatible with phased array antennas. {Finally,
		by utilizing narrowband signals with varying frequencies, \ac{far} systems can synthesize a large bandwidth with narrowband waveforms \cite{Huang2012,Yang2013}, \Revise{which simplifies the implementation of the waveform generator,  facilitates the receiver operation},  and allows the usage of non-linear amplifiers without limiting their power efficiency.}
	
	A major drawback of \ac{far} compared to wideband radar is its reduced range-Doppler reconstruction performance of targets.  This reduced performance is a byproduct of the relatively small number of radar measurements processed by \ac{far}, which stems from its usage of a single narrowband waveform for each pulse. 
	The performance reduction can be relieved by using \ac{cs} algorithms that exploit sparsity of the target scheme \cite{Huang2018}. 
	However, the degradation becomes notable in extremely congested or contested electromagnetic environments \cite{Huang2018a}, where there may be no vacant bands in some pulses or some radar returns of the transmitted pulses may be discarded due to strong interference \cite{Wang2006,Rao2011}.  
	
	The performance degradation of \ac{far} can be mitigated by using multi-carrier transmissions. When  multiple carriers are  transmitted simultaneously in a single pulse, the number of radar measurements is increased, and  the target reconstruction performance is improved. 
	Various multi-carrier radar schemes have been studied in the literature, including \ac{fdmamimo} \cite{Sun2014,cohen2019}, \ac{summer} \cite{Cohen2018}, and \ac{fda} radar \cite{Antonik2006, Liu2017}. In the \Revise{aforementioned} schemes, different array elements transmit waveforms at different frequencies, usually forming an omnidirectional beam and illuminating a large field-of-view\cite{Eli2013}. This degrades radar performance, especially in track mode, where a highly directional beam focusing on the target is preferred\cite{Eli2013}. In addition, frequency agility is not exploited in \ac{fdmamimo} and \ac{fda}. The derivation of frequency agile multi-carrier schemes for phased array radar, which leads to a focused  beam with high gain, is the focus of this work.

	Here, we propose two multi-carrier agile phased array radar schemes. The first uses all the antenna elements to transmit a single waveform consisting of multiple carriers simultaneously in each pulse.  Frequency agility is induced by randomly selecting the carriers utilized, resulting in a \ac{wmar} scheme. While the increased number of carriers is shown to achieve improved reconstruction performance compared to conventional \ac{far} \cite{Huang2018a}, \ac{wmar} utilizes multiband signals \Revise{of large instantaneous bandwidth. Therefore},  its implementation \Revise{does not benefit from the simplifications associated with utilizing}  conventional narrowband monotone \Revise{signals}, and may suffer from envelope fluctuation \cite{Chen2017multiband}.  
	
	To overcome the \Revise{use of instantaneously} wideband waveforms, we next develop  \ac{caesar}, which combines frequency agility and {\em spatial agility}. Specifically, \ac{caesar} selects a small number of carrier frequencies on each pulse and randomly allocates different carrier frequencies among its antenna elements, such that each array element transmits a narrowband constant modulus waveform\Revise{, facilitating system implementation}. 
	An illustration of this transmission scheme is depicted in Fig.~\ref{fig:multi-carrier}. 
	
	For each carrier frequency, \Revise{dedicated} phase shifts on the corresponding sub-array elements are used to yield a directional transmit beam, allowing to illuminate the tracked target in a similar manner as phased array radar. Despite the fact that only a sub-array antenna is utilized for each frequency, the antenna-frequency hopping strategy of \ac{caesar} results in array antenna gain loss and a relatively small performance gap  compared to wideband radar equipped with the same antenna array. Furthermore, the combined randomization of frequency and antenna allocation can be exploited to realize a \ac{dfrc} system \cite{Hassanien2016a,Sturm2011, Ma2018,Ma2019} by embedding digital information into the selection of these parameters. We study the application of \ac{caesar} as a \ac{dfrc} system in a companion paper \cite{Huang2019}, and focus here on the radar and its performance.

	\begin{figure}
		\centering
		\includegraphics[width=\myfiguresize in]{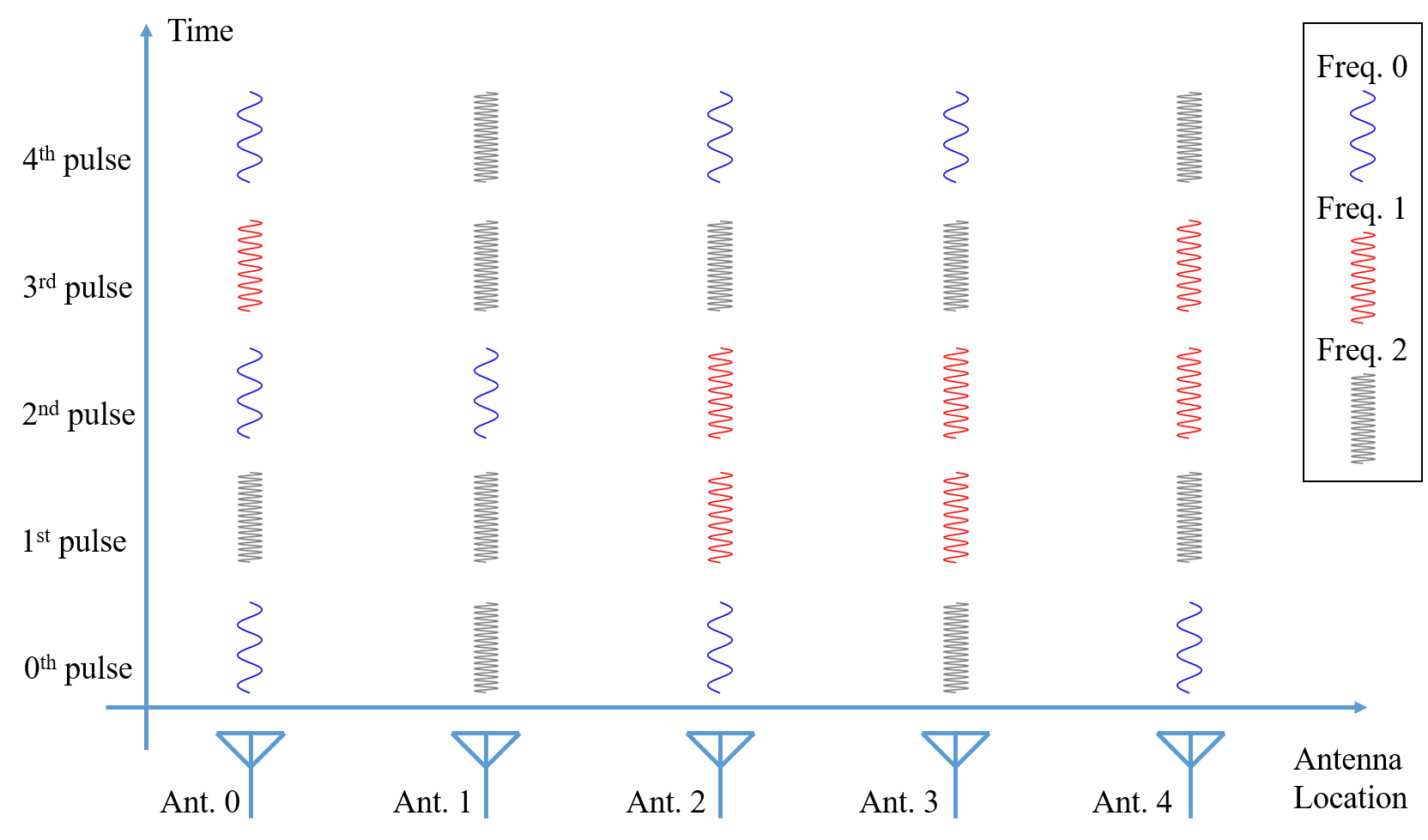}
		\caption{Transmission example of \ac{caesar}. In every pulse of this example, two out of three carrier frequencies are emitted by different sub-arrays. For example, frequency 0 and 2 are selected in the 0-th pulse and are sent by antenna 0, 2, 4 and antenna 1, 3, respectively. \ac{far} or \ac{fdmamimo}/\ac{fda} can be regarded as a special case of \ac{caesar}, with only one out of three frequencies or all available frequencies sent in each pulse.}
		\vspace{-0.6cm}
		\label{fig:multi-carrier}
	\end{figure}

	To present \ac{wmar} and \ac{caesar}, we characterize the signal model for each approach, based on which we develop a recovery algorithm for \ac{hrr}, Doppler, and angle estimation of radar targets. Our proposed algorithm utilizes \ac{cs} methods for range-Doppler reconstruction, exploiting its underlying sparsity, and applies matched filtering to detect the angles of the targets.  We provide a detailed theoretical analysis of the range-Doppler recovery performance of our proposed algorithm under complex electromagnetic environments. In particular, \Revise{we prove} that \ac{caesar} and \ac{wmar} are guaranteed to recover with high probability a number of scattering points which grows proportionally to the square root of the number of different narrowband signals used, i.e., the number of carrier frequencies that are simultaneously transmitted in each pulse. This theoretical result verifies that increasing the number of carriers improves target recovery, and reduces performance degradation due to intense interference in complex electromagnetic environments.
	\ac{wmar} and \ac{caesar} are evaluated in a numerical study, where it is shown that their range-Doppler reconstruction performance as well as robustness to interference are substantially improved compared to  \ac{far}. Additionally, it is demonstrated that the performance of \ac{caesar} is only within a small gap from that achievable using  wideband \ac{wmar}.

	\par

	The remainder of the paper is structured as follows: Sections \ref{sec:WMAR} and \ref{sec:system} present   \ac{wmar} and \ac{caesar}, respectively. 
	Section~\ref{sec:signalprocessing} introduces the recovery algorithm to estimate the range, Doppler, and angle of the targets. 
	In Section~\ref{sec:comparison} we discuss the pros and cons of each scheme compared to related radar methods. Section~\ref{sec:analysis} derives theoretical performance measures of the recovery method. Simulation results are presented in Section~\ref{sec:sim}, and Section \ref{sec:conclusion} concludes the paper.
	
	Throughout the paper, we use   $\mathbb{C}$, $\mathbb{R}$ to denote the sets of complex, real numbers, respectively, and use $| \cdot |$ for the magnitude or cardinality of a scalar number, or a set, respectively. 
	Given $x \in \mathbb{R}$, $\lfloor x \rfloor$ denotes the largest integer less than or equal to $x$, and $\binom{n}{k} = \frac{n!}{k!(n-k)!}$ represents the binomial coefficient. 
	Uppercase boldface letters denote matrices (e.g., $\bm A$), and boldface lowercase letters  denote vectors (e.g., $\bm a$).
	The $(n,m)$-th element of \Revise{a} matrix $\bm A$ is denoted as $[{\bm A}]_{m,n}$, and similarly $[{\bm a}]_{n}$ is the $n$-th entry of the vector $\bm a$.  Given a matrix $\bm A \in \mathbb{C}^{M \times N}$, and a number $n$ (or a set of integers, $\Lambda$), $\left[\bm A\right]_n$ ($\left[\bm A\right]_{\Lambda} \in \mathbb{C}^{M \times |\Lambda|}$) is the $n$-th column of $\bm A$ (the sub-matrix consisting of the columns of $\bm A$ indexed by $\Lambda$). 
	Similarly,  $\left[\bm a\right]_{\Lambda} \in \mathbb{C}^{ |\Lambda|}$ is the sub-vector consisting of the elements of $\bm a \in \mathbb{C}^{N}$ indexed by $\Lambda$. 
	The complex conjugate, transpose, and the complex conjugate-transpose are denoted $(\cdot)^*$, $(\cdot)^T$, $(\cdot)^H$, respectively. 
	\Revise{We denote} $\| \cdot \|_p$ as the $\ell_p$ norm,  $\| \cdot \|_0$ is the number of non-zero entries, and  $\| \cdot \|_F$ is the Frobenius norm. 
	The probability measure is $\mathbb{P}(\cdot)$, while ${\rm E}[\cdot]$ and ${\rm D}[\cdot]$ are the expectation and variance of a random argument, respectively. 

	\vspace{-0.2cm}
	\section{WMAR}
	\label{sec:WMAR}
	\vspace{-0.1cm}
	In this section we present the proposed \ac{wmar} scheme, which originates from \ac{far}  \cite{Axelsson2007}, aiming to increase the number of radar measurements and improve the range-Doppler recovery performance. 	
	We first briefly review \ac{far} in Subsection \ref{subsec:far}. Then, we detail the proposed \ac{wmar} in Subsection \ref{subsec:WMARmulti-carrier}, and present the resulting radar signal model in Subsection \ref{subsec:WMARreceive}.  
	
	\vspace{-0.2cm}
	\subsection{Preliminaries of FAR}
	\label{subsec:far}
	\vspace{-0.1cm}
	\ac{far} \cite{Axelsson2007} is a technique for enhancing the \ac{eccm} and \ac{emc} performance of radar systems by using randomized carrier frequencies. In the following we consider a radar system equipped with $L$ antenna elements,  uniformly located on an antenna array with distance $d$ between two adjacent elements. 
	Let $N$ be the number of radar pulses transmitted in each \ac{cpi}. Radar pulses are repeatedly transmitted, starting from time instance $nT_r$ to $nT_r + T_p$, $n \in \mySet{N} :=\! \{0,1,\ldots,N\!-\!1\}$, where $T_r$ and $T_p$ represent the pulse repetition interval and pulse duration, respectively, and $T_r > T_p$. 
	Let  $\mathcal{F}$ be the set of available carrier frequencies,  given by 
	$	\mathcal{F}:=\{ f_c + m\Delta f | m \in \mySet{M}\}$,
	where $f_c$ is the initial carrier frequency, $\mySet{M}: =\! \{0,1,\dots,M\!-\!1\}$,  $M$ is the number of available frequencies, and $\Delta f$ is the frequency step. 
	
	In the $n$-th radar pulse, \ac{far} randomly selects a carrier frequency 
	$f_n$ from $\mathcal{F}$.  The waveform sent from each antenna for the $n$-th pulse at time instance $t$ is  $\phi(f_{n}, t-nT_r)$, where 
	\begin{equation}
	\phi(f, t) :=  {\rm rect}\left({t}/{T_p} \right) e^{j2\pi f t}, 
	\end{equation}
	and ${\rm rect}(t) = 1$ for $t \in [0,1)$ and zero otherwise, representing rectangular envelope baseband signals.

	In order to direct the antenna beam pointing towards a desired angle $\theta$, the signal transmitted by each antenna is weighted by a phase shift $w_l(\theta,f_{n})\in \mathbb{C}$ \cite{Pillai1989}, given by
	\begin{equation}
	\label{eq:weight}
	w_l(\theta,f) := e^{j2\pi f ld \sin \theta/c},
	\end{equation}
	where $c$ denotes the speed of light. Define the vector $\bm w\left(\theta, f\right) \in \mathbb{C}^{L}$ whose $l$-th entry is $\left[\bm w\left(\theta,f\right)\right]_l := w_l\left(\theta,f\right)$. 
	The transmitted signal  can  be written as
	\begin{equation}
	\label{eq:xFAR}
	{\bm x}_{\rm F}(n,t)  := \bm w(\theta,f_{n}) \phi(f_{n}, t-nT_r).
	\end{equation}
	The vector ${\bm x}_{\rm F}(n,t) \in \mathbb{C}^{L}$ in \eqref{eq:xFAR} denotes the transmission vector of the full array for the $n$-th pulse at time instance $t$. 
	
	The fact that \ac{far} transmits monotone waveform facilitates its  realization. Furthermore,  the frequency agility achieved by randomizing the  frequencies between pulses enhances  survivability in complex electromagnetic environments. However, this comes at the cost of reduced number of radar measurements, which degrades the target recovery performance, particularly in the presence of  interference, where some of the radar returns are missed \cite{Huang2018a}. To overcome these drawbacks, in the following we propose \ac{wmar}, which extends \ac{far} to multi-carrier transmissions.

	\vspace{-0.2cm}
	\subsection{WMAR Transmit Signal Model}
	\label{subsec:WMARmulti-carrier}
	\vspace{-0.1cm}
	\ac{wmar} extends \ac{far} to multi-carrier signalling. Broadly speaking, \ac{wmar} transmits a single multiband waveform from all its antennas, maintaining  frequency agility by randomizing a {\em subset} of the available frequencies on each pulse. 
	
	Specifically, in the $n$-th radar pulse, \ac{wmar} randomly selects a set of carrier frequencies 
	$\mySet{F}_n$ from $\mathcal{F}$, $\mySet{F}_n \subset \mathcal{F}$. We assume that the cardinality of $\mySet{F}_n$ is constant, i.e., $|\mySet{F}_n| = K$ for each $n \in \mySet{N}$, and write the elements of this set as $\mySet{F}_n = \{\Omega_{n,k}| k \in \mySet{K} \}$, $\mySet{K}:=\{0,1,\dots, K-1\}$. The portion of the  $n$-th pulse of \ac{wmar} in the $k$-th frequency is given by ${\bm x}_{{\rm W}, k}(n,t) := \frac{1}{\sqrt{K}} \bm w\left(\theta,\Omega_{n,k}\right) \phi\left(\Omega_{n,k}, t-nT_r\right)$, and the overall transmitted vector is $\bm x_{\rm W}(n,t)  \! =\! \sum_{k=1}^{K}{\bm x}_{{\rm W}, k}(n,t) $, i.e.,
	\vspace{-0.1cm}
	\begin{align}
	\bm x_{\rm W}(n,t)  
	&=  \sum\limits_{k=1}^{K}\frac{1}{\sqrt{K}} \bm w\left(\theta,\Omega_{n,k}\right) \phi\left(\Omega_{n,k}, t-nT_r\right),
	\label{eq:xWMAR}
	\vspace{-0.1cm}
	\end{align}
	where the factor $\frac{1}{\sqrt{K}}$ guarantees that \eqref{eq:xWMAR} has the same total power as the \ac{far} signal \eqref{eq:xFAR}. 
	
	\ac{far} is a special case of \ac{wmar} under the setting $K=1$. By using multiple carriers simultaneously via wideband signalling, i.e., $K>1$, \ac{wmar} transmits a highly directional beam, while improving the robustness to missed pulses compared to \ac{far}. The improved performance stems from the use of multi-carrier transmission, which increases the number of radar measurements. To see this, we detail the received signal model of \ac{wmar} in the following subsection.
	
	%
	%
	\subsection{WMAR Received Signal Model}
	\label{subsec:WMARreceive}
	We next model the received signal processed by \ac{wmar} for target identification. To that aim, we focus on the time interval after the $n$-th pulse is transmitted, i.e., $nT_r + T_p < t< (n+1)T_r$. In this period, the radar receives echoes of the  pulse, which are sampled and processed in discrete-time. 
	
	To formulate the radar returns, we assume an ideal scattering point\Revise{, representing either target or clutter,} with scattering coefficient $\beta \in \mathbb{C}$ located in the transmit beam of the radar with direction angle $\vartheta$, i.e., $\vartheta \approx \theta$. Denote by $r(t)$ the range between the target\Revise{/clutter} and the first radar antenna array element at time $t$. The scattering point is moving at a constant velocity $v$ radially along with the radar line of sight, i.e., $r(t) = r(0) + vt$. Under the ``stop and hop'' assumption \cite[Page 99, Ch. 2]{richards2005fundamentals}, which assumes that the target hops to a new location when the radar transmits a pulse and stays there until another pulse is emitted, the range in the $n$-pulse is approximated as 
	\vspace{-0.1cm}
	\begin{equation}
	r(t) \!\approx\! r(nT_r) \!=\! r(0) \!+\! v \cdot nT_r, \quad nT_r\! <\!t\!<\! (n\!+\!1)T_r.
	\vspace{-0.1cm}
	\label{eqn:rtdef}
	\end{equation}
	
	To model the received signal, we first consider the $n$-th radar pulse that reaches the target, denoted by  ${\tilde{x}}(n,t) $. Let ${\tilde{x}}_k(n,t)$ be its component at frequency $\Omega_{n,k}$, i.e., ${\tilde{x}}(n,t) := \sum_{k=0}^{K-1}{\tilde{x}}_k(n,t)$. Note that ${\tilde{x}}_k(n,t)$ is a summation of delayed transmissions from the corresponding antenna elements. The delay for the $l$-th array element is ${r(nT_r)}/{c}+{ld\sin \vartheta}/{c}$. Under the narrowband, far-field assumption, using \eqref{eq:weight}, we have that
	\vspace{-0.1cm}
	\begin{eqnarray}
	{\tilde{x}}_k(n,t)
	\!\!\!\!\!\!
	&=&	\!\!\! \sum \limits_{l=0}^{L-1} \left[ \bm x_{{\rm W},k}(n,t-{r(nT_r)}/{c})\right]_l e^{-j2\pi \Omega_{n,k}{ld\sin \vartheta}/{c} }\nonumber \\
	\!\!\!\!\!\!& {=}&\!\!\! \bm w^H\left(\vartheta,\Omega_{n,k}\right) \bm x_{{\rm W},k}(n,t-{r(nT_r)}/{c}). \label{eq:signal_at_scatterer}
	\end{eqnarray}
	Substituting \eqref{eqn:rtdef} and the definition of $\bm x_{{\rm W},k}(n,t)$  into \eqref{eq:signal_at_scatterer} yields
	\vspace{-0.1cm}
	\begin{equation}
	\!\!{\tilde{x}}_k(n,t)\! =\! \frac{{\rho}_{\rm W} (n,k,\delta_{\vartheta})}{\sqrt{K}}\phi\!\left(\Omega_{n,k}, t\!-\!nT_r\!-\!\frac{r(0)\!+\!nvT_r}{c}\right)\!,
	\label{eqn:xktempW}
	\end{equation}
	where $\delta_{\vartheta} :=  \sin \vartheta \!- \!\sin \theta$ is the relative direction sine with respect to the  transmit beam, and ${\rho}_{\rm W} (n,k,\delta_{\vartheta}):=\bm w^H\left(\vartheta,\Omega_{n,k}\right) \bm w\left(\theta,\Omega_{n,k}\right)$ is the transmit gain, expressed  as
	\vspace{-0.1cm}
	\begin{equation}
	{\rho}_{\rm W} (n,k,\delta_{\vartheta}) = \sum_{l=0}^{L-1}{ e^{-j 2\pi \Omega_{n,k} ld\delta_{\vartheta}/{c}}}. 
	\label{eq:rhoW}
	\vspace{-0.1cm}
	\end{equation}
	Note that $\rho_{\rm W} (n,k,\delta_{\vartheta})$ approaches $L$ when $\delta_{\vartheta} \approx 0$.
	
	Having modeled the signal which reaches the target, we now derive the radar returns observed by the antenna array. After being reflected by the scattering point, the signal at the $k$-th frequency propagates back to the $l$-th radar array element with an extra delay of $r(nT_r)/c + {ld\sin \vartheta}/{c}$, resulting in
	\begin{equation}
	\left[\widetilde{\bm y}_{{\rm W},k}(n,t)\right]_l:=\beta {\tilde{x}}_k\left(n,t-{r(nT_r)}/{c}-{ld\sin \vartheta}/{c}\right) .
	\end{equation}
	The echoes vector $\widetilde{\bm y}_{{\rm W},k}(n,t)\in \mathbb{C}^L$ can be written as
	\begin{align}
	\widetilde{\bm y}_{{\rm W},k}(n,t)\!&=\!\beta {\bm w}^*\left(\vartheta,\Omega_{n,k}\right) {\tilde{x}}_k\left(n,t-{r(nT_r)}/{c}\right)  \notag\\
	\!&\stackrel{(a)}{=}\! \frac{\beta}{\sqrt{K}} {\bm w}^*\left(\vartheta,\Omega_{n,k}\right)\rho_{\rm W} (n,k,\delta_{\vartheta})\notag \\ 
	&\!\times\phi\left(\Omega_{n,k}, t\!-\!nT_r\!-\!{(2r(0)+2nvT_r)}/{c}\right),
	\label{eqn:RadarRx1W}
	\end{align}
	where $(a)$ follows from \eqref{eqn:xktempW}.
	
	The received echoes at all $K$ frequencies are then separated and sampled independently by each array element. 	The  signal $\widetilde{\bm y}_{{\rm W},k}(n,t)$ is sampled with a rate of $f_s = 1/T_p$ at time instants $t =  nT_r + i/f_s$, $i=0,1,\dots,\lfloor T_r f_s\rfloor-1$, such that each pulse is sampled once. Every sample time instant corresponds to a \ac{crc}, $r \in \left( \frac{i-1}{2f_s}c, \frac{i}{2f_s}c\right)$. 
	The division to \ac{crc}s indicates coarse range information of scattering points. 
	We focus on an arbitrary $i$-th \ac{crc}, assuming that the scattering point does not move between \ac{crc}s during a \ac{cpi}, i.e., there exists some integer $i$ such that
	\begin{equation}
	\left\{ \begin{array}{l}
	\frac{i-1}{2f_s}c<r\left( 0 \right) <\frac{i}{2f_s}c,\\
	\frac{i-1}{2f_s}c<r\left( 0 \right) +vnT_r<\frac{i}{2f_s}c, \quad \forall n \in \mySet{N}.\\
	\end{array} \right.
	\end{equation}
	
	Collecting radar returns from $N$ pulses and $L$  elements at the same \ac{crc} yields a data cube $\bm Y_{\rm W} \in \mathbb{C}^{L \!\times\! N\! \times\! K}$ with entries 
	\begin{equation}
	\label{eq:data_cube}
	[\bm Y_{\rm W}]_{l,n,k}:= 	\left[\widetilde{\bm y}_{{\rm W},k}(n,n T_r + i/f_s)\right]_l,   
	\end{equation}
	where $i$ is the \ac{crc} index. 
	The data cube $\bm Y_{\rm W}$ is processed to estimate the refined range information, Doppler, and angle of the scattering point. Data cubes from different \ac{crc}s are processed identically and separately.
	
	Finally, we formulate how the unknown parameters of the targets are embedded in the processed data cube $\bm Y_{\rm W}$. To that aim, define $\delta_r := r(0)-ic/2f_s$ as the high-range resolution distance,  $c_{n,k} := (\Omega_{n,k}-f_c)/\Delta f \in \mySet{M}$ as the carrier frequency index, and $\zeta_{n,k} = \Omega_{n,k}/f_c$ as the relative frequency factor. Then, denoting by $\tilde{\beta} := \beta e^{-j4\pi f_c \delta_r/c}$, $\tilde{r} := -4\pi \Delta f \delta_r/c$ and $\tilde{v} := -4\pi f_c v T_r/c$ the generalized scattering intensity, and the normalized range and velocity, respectively, and substituting \eqref{eqn:RadarRx1W} into \eqref{eq:data_cube}, we have that
	\vspace{-0.1cm}
	\begin{align}
	\!\!  [\bm Y_{\rm W}]_{l,n,k} \!=\!\! 	\frac{\tilde{\beta}e^{j \tilde{r} c_{n,k}}}{\sqrt{K}}& e^{j \tilde{v} n \zeta_{n,k}}e^{\!-\!j 2\pi  \frac{\Omega_{n,k} ld {\sin \vartheta}}{c}} \! { \rho}_{\rm W} (n,k,\delta_{\vartheta}).  
	\label{eq:data_cube2} 
	\vspace{-0.1cm}
	\end{align}
	The unknown parameters in \eqref{eq:data_cube2} are $\tilde{\beta} $, $\tilde{r}$, $\tilde{v}$ and $(\sin \vartheta, \delta_{\vartheta})$, which are used to reveal the scattering intensity $|\beta|$, \ac{hrr} range $r(0)$, velocity $v$ and angle $\vartheta$ of the target.
	
	The above model can be naturally extended to noisy multiple scatterers. When there are $S$ scattering points inside the \ac{crc} instead of a single one as assumed previously, the received signal is a summation of returns from all these points corrupted by additive noise, denoted by $\bm N \in \mathbb{C}^{L \times N \times K}$. Following \eqref{eq:data_cube2}, the entries of the data matrix are  
	\vspace{-0.1cm}
	\begin{align} 
	\left[ \bm Y_{\rm W}\right]_{l,n,k} \!=\! \frac{1}{\sqrt{K}}\sum \limits_{s = 0}^{S-1}&\tilde{\beta}_s e^{j \tilde{r}_s c_{n,k}}e^{j \tilde{v}_s n \zeta_{n,k}}e^{-j 2\pi \Omega_{n,k} ld {\sin \vartheta_s}/c} \notag \\
	&\times{ \rho}_{\rm W} (n,k,\delta_{\vartheta_s}) + \left[\bm N\right]_{l,n,k}, 
	\label{eq:entry_y_sumW}
	\vspace{-0.1cm}
	\end{align}
	where  $\{\tilde{\beta}_s\}$, $\{\tilde{r}_s\}$, $\{\tilde{v}_s\}$ and $\{\vartheta_s\}$ represent the sets of factors  of scattering coefficients, ranges, velocities, and angles of the $S$ scattering points, respectively, which are unknown and should be estimated.  A method for recovering these parameters from the data cube $ \bm Y_{\rm W}$ is detailed in Section \ref{sec:signalprocessing}. 
	
	\ac{wmar} has several notable advantages: First, as an extension of \ac{far}, it preserves its frequency agility and is suitable for implementation with phased array antennas. Furthermore, as we discuss in Section \ref{sec:comparison}, its number of radar measurements for each \ac{crc} is increased by a factor of $K$ compared to \ac{far}, thus yielding increased robustness to interference. However,  \ac{wmar}  transmitters simultaneously send multiple carriers instead of a monotone as in \ac{far}, which requires large instantaneous bandwidth, leading to envelope fluctuation and low amplifier efficiency. To overcome these issues, we introduce \ac{caesar} in the following section, which utilizes narrowband radar transceivers while introducing spatial agility, enabling multi-carrier transmission using monotone signals at a cost of a minimal array antenna gain loss.
	
	\section{CAESAR}
	\label{sec:system}
	\ac{caesar},  similarly to \ac{wmar}, extends \ac{far} to multi-carrier transmission. However, unlike \ac{wmar}, \ac{caesar} utilizes monotone signalling and reception, and is thus more suitable for implementation. We detail the transmit and receive models of \ac{caesar} in Subsections \ref{subsec:multi-carrier} and \ref{subsec:receive}, respectively.

	\subsection{CAESAR Transmit Signal Model}
	\label{subsec:multi-carrier}
	%
	Broadly speaking, \ac{caesar} extends \ac{far} to multi-carrier signalling by transmitting monotone waveforms with varying frequencies from different antenna elements. The selection of the frequencies, as well as their allocation among the antenna elements, is randomized anew in each pulse, thus inducing both {\em frequency and spatial agility}.	
	
	To formulate \ac{caesar}, we consider the same pulse radar formulation detailed in Section \ref{sec:WMAR}.	Similarly to \ac{wmar} detailed in Subsection \ref{subsec:WMARmulti-carrier}, in the $n$-th radar pulse, \ac{caesar} randomly selects a set of carrier frequencies 
	$\mySet{F}_n = \{\Omega_{n,k}| k \in \mySet{K} \}$ from $\mathcal{F}$. While \ac{wmar} uses the set of selected frequencies to generate wideband waveforms, \ac{caesar} allocates a sub-array for each frequency, such that all the antenna array elements are utilized for transmission, each at a single carrier frequency.
	Denote by $f_{n,l} \in \mySet{F}_n$ the frequency used by the $l$-th antenna array element, $l \in \mySet{L}:=\{0,1,\dots, L-1\}$. After phase shifting the waveform to direct the beam, the $l$-th array element transmission  can  be written as
	\vspace{-0.1cm}
	\begin{equation}
	\label{eq:x}
	\left[ \bm x_{\rm C}(n,t) \right]_l := \left[ \bm w(\theta,f_{n,l}) \right]_l \phi(f_{n,l}, t-nT_r).
	\vspace{-0.1cm}
	\end{equation}
	The vector $\bm x_{\rm C}(n,t) \in \mathbb{C}^{L}$ in \eqref{eq:x} denotes the full array transmission vector for the $n$-th pulse at time $t$. Here, unlike \ac{far} which transmits a single frequency from the full array  \eqref{eq:xFAR}, \ac{caesar} assigns  diverse frequencies to different sub-array antennas, as illustrated in Fig.~\ref{fig:multi-carrier}.

	The transmitted signal \eqref{eq:x} can also be expressed by grouping the array elements which use the same frequency $\Omega_{n,k}$. Let ${\bm x}_{{\rm C},k}(n,t)\in \mathbb{C}^{L}$ with zero padding represent the portion of ${\bm x}_{\rm C}(n,t)$ which utilizes $\Omega_{n,k}$, i.e., 
	\vspace{-0.1cm}
	\begin{equation}
	\bm x_{{\rm C},k}(n,t)  = \bm P(n,k) \bm w\left(\theta,\Omega_{n,k}\right) \phi\left(\Omega_{n,k}, t-nT_r\right),
	\label{eqn:TxSignal_k}
	\vspace{-0.1cm}
	\end{equation}
	where $\bm P(n,k) \in \{0,1\}^{L \times L}$ is a diagonal selection matrix with diagonal $\bm p(n,k)  \in \{0,1\}^{L}$, whose $l$-th entry is one if the $l$-th array element transmits at frequency $\Omega_{n,k}$ and zero otherwise, i.e., $\left[ \bm P(n,k) \right]_{l,l} =\left[ \bm p(n,k) \right]_{l} = 1$ and $\left[ \bm x_{{\rm C},k}(n,t) \right]_{l} = \left[ \bm x_{\rm C}(n,t) \right]_{l}$ when $f_{n,l} = \Omega_{n,k}$. 
	The transmitted signal is thus $\bm x_{\rm C}(n,t) := \sum_{k=0}^{K-1}\bm x_{{\rm C},k}(n,t)$, namely
	\vspace{-0.1cm}
	\begin{equation}
	\bm x_{\rm C}(n,t)  = \sum \limits_{k=0}^{K-1} \bm P(n,k) \bm w\left(\theta,\Omega_{n,k}\right) \phi\left(\Omega_{n,k}, t-nT_r\right).
	\label{eqn:TxSignal}
	\vspace{-0.1cm}
	\end{equation}
	Comparing \eqref{eqn:TxSignal} with \eqref{eq:xWMAR}, we find that each array element of \ac{caesar} transmits a single frequency with unit amplitude while in \ac{wmar} all $K$ frequencies with amplitudes scaled by a factor $1/\sqrt{K}$ are sent by each element.
	
	The diagonal selection matrices $\bm P(n,0),\dots,\bm P(n,K-1)$ uniquely describe the allocation of antenna elements for the $n$-th pulse.  \ac{caesar} transmission scheme implies that
	$\sum _{k=0}^{K-1}\bm P(n,k) = \bm I_{L}$,
	i.e., all the antenna elements are utilized for the transmission of the $n$-th pulse. The trace of $\bm P(n,k)$ represents the number of antennas using the $k$-th frequency. Without loss of generality, we assume that $L/K$ is an integer and ${\rm tr} \left(\bm P(n,k)\right) = L/K$, for each $n \in  \mySet{N}$ and $k \in \mySet{K}$.
	\par 
	Phased array \ac{far} and \ac{fda} \cite{Liu2017} are special cases of \ac{caesar} with $K= 1$ and $K = M = L$, respectively. A  fundamental difference between these radar schemes is the transmit beam pattern. In \ac{far}, the same carrier frequency is utilized by all the elements, i.e., $\Omega_{n,k}$ and $f_{n,l}$ are identical for each $k \in \mySet{K}$ and $l \in \mySet{L}$, respectively, resulting in highly directional beam. In \ac{fda}, all available frequencies are transmitted simultaneously and one frequency corresponds to a single antenna element, leading to an omnidirectional beam which degrades radar performance and is not suitable for target tracking \cite{Eli2013}.
	The proposed \ac{caesar} uses only a subset of the available frequencies in each pulse and multiple antenna elements share the same frequency, thus achieving a compromise radiation beam that only illuminates the desired angle. Despite the  gain loss in comparison with \ac{far} discussed in Section \ref{sec:comparison}, \ac{caesar} achieves improved range-Doppler reconstruction performance and increased robustness to interference, as numerically demonstrated in Section \ref{sec:sim}.
	

	%
	%
	\vspace{-0.2cm}
	\subsection{CAESAR Received Signal Model}
	\label{subsec:receive}
	\vspace{-0.1cm}
	We next model the received signal processed by \ac{caesar}. Unlike \ac{wmar}, in which each antenna receives and separates different frequency components, in \ac{caesar}, the $l$-th antenna element only receives radar returns at frequency $f_{n,l}$, and abandons other frequencies. This enables the use of narrowband receivers, simplifying the hardware requirements.  
	
	Note that the derivation of the signal component received at the $k$-th frequency in \eqref{eq:signal_at_scatterer},  $\tilde{x}_k(n,t)$,  does not depend on the specific radar scheme.  Here, substituting \eqref{eqn:TxSignal_k}  into \eqref{eq:signal_at_scatterer} yields
	\begin{equation}
	\!\!\!{\tilde{x}}_k(n,t)\! =\! \rho_{\rm C} (n,k,\delta_{\vartheta})\phi\!\left(\Omega_{n,k}, t\!-\!nT_r\!-\!\frac{r(0)\!+\!nvT_r}{c}\right)\!,
	\label{eqn:xktemp}
	\end{equation}
	where $\rho_{\rm C} (n,k,\delta_{\vartheta})\!:=\!\bm w^H\left(\vartheta,\Omega_{n,k}\right)\bm P(n,k) \bm w\left(\theta,\Omega_{n,k}\right)$ is the transmit gain of the selected sub-array antenna, expressed  as
	\begin{equation}
	\rho_{\rm C} (n,k,\delta_{\vartheta}) =\sum_{l=0}^{L-1}{ \left[\bm p(n,k)\right]_l e^{-j 2\pi \Omega_{n,k} ld\delta_{\vartheta}/{c}}}. 
	\label{eq:rho}
	\end{equation}
	Note that, \Revise{in contrast} to the transmit gain of \ac{wmar} in \eqref{eq:rhoW} which tends to $L$, $\rho_{\rm C} (n,k,\delta_{\vartheta})$ approaches $L/K$ when $\delta_{\vartheta} \approx 0$.
	By repeating the arguments in the derivation of \eqref{eqn:RadarRx1W}, the echo vector $\widetilde{\bm y}_{{\rm C},k}(n,t)\in \mathbb{C}^L$ can be written as
	\begin{eqnarray}
	\widetilde{\bm y}_{{\rm C},k}(n,t)\!\!\!\!\!\!&=& \!\!\! \beta {\bm w}^*\left(\vartheta,\Omega_{n,k}\right)\rho_{\rm C} (n,k,\delta_{\vartheta})\notag \\ 
	&&\!\!\!\times\phi\left(\Omega_{n,k}, t\!-\!nT_r\!-\!{(2r(0)+2nvT_r)}/{c}\right).
	\label{eqn:RadarRx1}
	\end{eqnarray}
	
	\ac{caesar} receives and processes impinging signals  by the corresponding elements of the antenna array. In particular, only a sub-array, whose elements are indicated by $\bm P(n,k)$, receives the impinging signal $\widetilde{\bm y}_{{\rm C},k}(n,t)$; the other array elements are tuned to other frequencies. The zero-padded received signal at the $k$-th frequency, denoted by ${\bm y}_{{\rm C},k}(n,t) \in \mathbb{C}^{L}$,  is thus ${\bm y}_{{\rm C},k}(n,t) := \bm P(n,k) \widetilde{\bm y}_{{\rm C},k}(n,t)$. 
	The full array received signal  is given by ${\bm y}_{\rm C}(n,t) := \sum _{k=0}^{K-1}{\bm y}_{{\rm C},k}(n,t)$.

	The observed signal ${\bm y}_{\rm C}(n,t)$ is sampled in a similar manner as detailed in Subsection \ref{subsec:WMARreceive}. Since \ac{caesar} processes a single frequency component per antenna element, the measurements from each \ac{crc} are collected together as a data matrix $\bm Y_{\rm C} \in \mathbb{C}^{L \times N}$, as opposed to a $L \times N \times K$ cube processed by \ac{wmar}.   By repeating the arguments used for obtaining \eqref{eq:data_cube2}, it holds that 
	\begin{equation}
	\left[ \bm Y_{\rm C}\right]_{l,n}\!\!\! =\!\! \tilde{\beta} e^{j \tilde{r} c_{n,k}}e^{j \tilde{v} n \zeta_{n,k}}e^{-j 2\pi \Omega_{n,k} ld {\sin \vartheta}/c}\rho_{\rm C} (n,k,\delta_{\vartheta}),
	\label{eq:entry_y_simp}
	\end{equation}
	which can be extended to account for multiple targets and noisy measurements as in \eqref{eq:entry_y_sumW}, i.e., 
	\begin{align} 
	\left[ \bm Y_{\rm C}\right]_{l,n} &\!=\! \sum \limits_{s = 0}^{S-1}\tilde{\beta}_s e^{j \tilde{r}_s c_{n,k}}e^{j \tilde{v}_s n \zeta_{n,k}}e^{-j 2\pi \Omega_{n,k}  \frac{ld\sin \vartheta_s}{c}}\notag \\
	&\quad\quad \times \rho_{\rm C} (n,k,\delta_{\vartheta_s}) + \left[\bm N\right]_{l,n},
	\label{eq:entry_y_sum}
	\end{align}	
	where $\bm N \in \mathbb{C}^{L \times N}$ is the additive noise. 
	In order to recover the unknown parameters from the acquired data matrix \eqref{eq:entry_y_sum}, in the following section we present a dedicated recovery scheme. 

	%
	%
	\vspace{-0.3cm}
	\section{Target Recovery Method}
	\label{sec:signalprocessing}
	\vspace{-0.1cm}
	Here, we present an algorithm for reconstructing  the unknown \ac{hrr} range, velocity, angle, and scattering intensity parameters of the scattering points from the radar measurements of both \ac{wmar} and \ac{caesar}. \Revise{{Detection is performed based on the estimated scattering intensities.} {The detected scattering points may belong to either target or clutter, and they are identified by their Doppler estimates. The motivation for this approach is that in many ground-based radar systems, fast moving targets are of interest, while static or slow moving scatterers with zero or nearly zero Doppler are regarded as clutter. We henceforth model the Doppler of both targets and clutter as unknown parameters, which are simultaneously estimated. In specific applications where the clutter Doppler is a-priori known, one can apply clutter mitigation \cite{Axelsson2006} in advance to the target recovery method. A similar procedure is also applied in pulse Doppler radars \cite[Ch. 5.5.1]{richards2005fundamentals}, where the moving target indication filtering for gross clutter removal is placed prior to the pulse Doppler filter bank.}}	 
	
	In order to maintain feasible computational complexity, we do not estimate all the parameters simultaneously: our proposed algorithm first jointly recovers the range-Doppler parameters followed by estimation of the unknown angles. When performing joint range-Doppler estimation, we assume that all the scattering points are located within the mainlobe of the transmit beam, and that the difference of the angle sine is negligible, i.e., $\delta_{\vartheta} \approx 0$. We then estimate the direction angles of scattering points based on their range-Doppler estimates. 
	
	We divide the target recovery method into three stages: 1) apply receive beamforming such that the magnitude of the received signal is enhanced, facilitating   range-Doppler \Revise{recovery}; 2) apply \ac{cs} methods for joint reconstruction of range and Doppler\Revise{, followed by a target detection procedure}; and 3) angle and scattering intensity estimation. These steps are discussed in Subsections \ref{subsec:beamform}-\ref{subsec:angle}, respectively. 
	A theoretical analysis of the range-Doppler estimation performance of our algorithm is provided in Section \ref{sec:analysis}, where we quantify how using multiple carriers improves the range-Doppler reconstruction performance.
	
	\vspace{-0.2cm}
	\subsection{Receive Beamforming}
	\label{subsec:beamform}
	\vspace{-.1cm}
	The first step in processing the radar measurements is to beamform the received signal in order to facilitate recovery of the range-Doppler parameters. This receive beamforming is applied to radar returns at different frequencies separately. 
	To formulate the beamforming technique, we henceforth focus on the $k$-th frequency of the $n$-th pulse, $\Omega_{n,k}$. For both \ac{caesar} and \ac{wmar}, a total of $L$ measurements correspond to $\Omega_{n,k}$, and are denoted \Revise{by} $\tilde{\bm z}_{n,k} \in \mathbb{C}^L$. For \ac{caesar}, $\tilde{\bm z}_{n,k}$ is given by 
	$\tilde{\bm z}_{n,k} =\bm P(n,k) \left[ \bm Y_{\rm C} \right]_n$,
	of which only elements corresponding to the selected sub-array are nonzero.
	For \ac{wmar}, $\tilde{\bm z}_{n,k}$ consists of the entries  $[\bm Y_{\rm W}]_{l,n,k}$ for each $l \in \mySet{L}$. 
	These measurements are integrated with the weights $\bm w\left( \theta, \Omega_{n,k}\right)$ such that the receive beam is pointed towards  $\theta$, resulting in 
	\vspace{-.1cm}
	\begin{equation}
	\label{eq:receive_beamforming_k}
	Z_{k,n} := \bm w^T\tilde{\bm z}_{n,k} \in \mathbb{C}.
	\vspace{-.1cm}
	\end{equation} 
	Define $\Gain := L^2/K^2$ for \ac{caesar}, and $\Gain := L^2/\sqrt{K}$ for \ac{wmar}.
	When $\delta_{\vartheta_s} \approx 0$, i.e., the beam direction $\theta$ is close to the true angle of the target, the resulting beam pattern can be simplified as stated in the following lemma:
	\begin{lemma}
		\label{lem:BeamPattern}
		If the difference of the angle sine satisfies $\delta_{\vartheta_s} \approx 0$, then $Z_{k,n}$ in \eqref{eq:receive_beamforming_k} can be approximated as
		\vspace{-0.1cm}
		\begin{equation}
		Z_{k,n} \approx  \Gain\sum \limits_{s = 0}^{S-1}\tilde{\beta}_s e^{j \tilde{r}_s c_{n,k}}e^{j \tilde{v}_s n \zeta_{n,k}}.
		\label{eq:signal_beamforming}
		\vspace{-0.1cm}
		\end{equation}  
	\end{lemma}
	
	\begin{proof} 
		See Appendix \ref{app:BeamPattern}.
	\end{proof}

	The receive beamforming  produces the matrix $\bm Z \in \mathbb{C}^{K \times N}$ whose entries are $[\bm Z]_{k,n} := Z_{k,n}$, for each $k \in \mySet{K}$, $n \in \mySet{N}$.
	Under the approximation \eqref{eq:signal_beamforming}, the obtained $\bm Z$ is used for range-Doppler reconstruction, as discussed in the next subsection.

	\vspace{-.2cm}
	\subsection{Range-Doppler Reconstruction Method}
	\label{subsec:RangeDoppler}
	\vspace{-.1cm}
	To reconstruct the range-Doppler parameters \Revise{ and detect targets in the presence of noise and/or clutter}, we first recast the beamformed signal model of Lemma \ref{lem:BeamPattern} in matrix form, and then apply \ac{cs} methods to recover the unknown parameters, exploiting the underlying sparsity of the resulting model. \Revise{The targets of interest are then identified based on the estimated parameters.}
	\par
	To obtain a sparse recovery problem, we start by discretizing the range and Doppler domains.
	Recall that $\tilde r_s$ and ${\tilde v}_s$ denote the normalized range and Doppler parameters, with resolutions $\frac{2\pi}{M}$ and $\frac{2\pi}{N}$, corresponding to the numbers of available frequencies and pulses, respectively. Both parameters belong to continuous domains in the unambiguous region $\left({\tilde r}_s,{\tilde v}_s\right) \in [0,2\pi)^2$. We discretize ${\tilde r}_s$ and ${\tilde v}_s$ into \ac{hrr} and Doppler grids, denoted by grid sets $\mathcal{R} := \left\{  \left.\frac{2\pi m}{M}  \right| m \in \mySet{M} \right\}$ and $\mathcal{V} := \left\{  \left.\frac{2\pi n}{N}  \right| n \in \mySet{N} \right\}$, \Revise{with grid intervals $ \Delta_{\tilde{r}} = \frac{2\pi}{M} $ and  $ \Delta_{\tilde{v}} = \frac{2\pi}{N} $, }respectively, and assume that the targets are located precisely on the grids. 
	The target scene
	can now be represented by the matrix ${\bm B} \in \mathbb{C}^{M \times N}$ with entries
	\begin{equation}
	[\bm B]_{m,n}:=\left\{ \begin{array}{l}
	\tilde{\beta}_s \Gain, \text{ if } \text{ } \left({\tilde r}_s, {\tilde v}_s \right) = \left( \frac{2\pi m}{M}, \frac{2\pi n}{N}\right),\\
	0,\text{ otherwise}.\\
	\end{array} \right. 
	\label{eqn:BmatForm}
	\end{equation}

	We can now use the sparse structure of \eqref{eqn:BmatForm} to formulate the range-Doppler reconstruction as a sparse recovery problem. To that aim, let $\bm z \in \mathbb{C}^{KN}$ and  ${\bm \beta}\in \mathbb{C}^{MN}$ be the vectorized representations of  $\bm Z$ and $\bm B$, respectively, i.e.,  $[\bm z]_{k+nK} = Z_{k,n}$ and  $[\bm \beta]_{n+mN} := [\bm B]_{m,n}$. 
	From \eqref{eq:signal_beamforming}, it holds that
	\begin{equation}{\label{eq:noiselessSigModel}}
	{\bm z} ={\bm \Phi}{\bm \beta},
	\end{equation}
	where the entries of $\bm \Phi \in \mathbb{C}^{KN \times MN}$ are given by
	\begin{equation}\label{eq:phi}
	\left[{\bm \Phi} \right]_{k\!+\!nK,l\!+\!mN}\! :=\! e^{j\frac{2\pi m}{M}c_{n,k}\!+\! j\frac{2\pi l}{N} n\zeta_{n,k}}, 
	\end{equation} 
	$m \!\in\! \mySet{M}$, $l,n \!\in\! \mySet{N}$, and $k \in \mySet{K}$.
	The matrix $\bm \Phi$ is determined by the frequencies utilized in each pulse. Consequently, $\bm \Phi$  is a random matrix, as these parameters are randomized by the radar transmitters, whose realization is known to the receiver. 
	
	In the presence of noisy radar returns, \eqref{eq:noiselessSigModel}  becomes
	\begin{equation}{\label{eq:noisySigModel}}
	{\bm z} ={\bm \Phi}{\bm \beta} +{\bm n},
	\end{equation}
	where the entries of the noise vector $\bm n \in \mathbb{C}^{KN}$ are the beamformed noise, e.g., for \ac{caesar} these are given by $\left[ \bm n \right]_{k+nK} =  \bm w^T\left( \theta, \Omega_{n,k}\right)\bm P(n,k) \left[ \bm N \right]_n$.
	\par
	Since in each pulse only a subset of the available frequencies are transmitted, i.e., $K \le M$, the sensing matrix ${\bm \Phi}$ in (\ref{eq:noisySigModel}) has more columns than rows, $MN \geq KN$, indicating that solving \eqref{eq:noisySigModel} is naturally an under-determined problem. When $\bm \beta$ is $S$-sparse, which means that there are \Revise{at most} $S$ non-zeroes in $\bm \beta$ and $S \ll MN$, \ac{cs} algorithms can be used to solve \eqref{eq:noisySigModel}, \Revise{yielding the estimate $\hat{\bm \beta}$.} 
	
	{\Revise{Particularly,} 	\ac{cs} methods aim to solve under-determined problems such as \eqref{eq:noiselessSigModel} by seeking the sparsest solution, i.e.,  
		\vspace{-0.1cm}
		\begin{equation}
		\label{eq:ell0}
		\hat{\bm \beta} = \mathop{\arg \min }_{\bm \beta} \left\| \bm \beta \right\|_0, {\text{ s.t. }} \bm z = \bm \Phi \bm \beta.
		\vspace{-0.1cm}
		\end{equation}
		The  $\ell_0$ optimization in \eqref{eq:ell0} is generally NP-hard. To reduce computational complexity, many alternatives including $\ell_1$ optimization and greedy approaches have been suggested to approximate \eqref{eq:ell0}, see \cite{Eldar2012, Eldar2015}.
		\par
		We take $\ell_1$ optimization as an example\Revise{, under which we provide a theoretical analysis and numerically evaluate the performance in Sections \ref{sec:analysis} and \ref{sec:sim}, respectively.} \Revise{In particular, in} the absence of noise, \Revise{we use the} basis pursuit algorithm, \Revise{which} solves
		\begin{equation}
		\label{eq:l1}
		\hat{\bm \beta} = \mathop{\arg \min }_{\bm \beta} \left\| \bm \beta \right\|_1, {\text{ s.t. }} \bm z = \bm \Phi \bm \beta,
		\end{equation}
		instead of \eqref{eq:ell0}. {\Revise{In noisy cases, recovering $\bm \beta$ can be formulated as minimizing the $ \ell_1 $ norm under a $ \ell_2 $ constraint on the fidelity:
				\begin{equation}
				\label{eq:l1bpdn}
				\hat{\bm \beta} = \mathop{\arg \min }_{\bm \beta}  \left\| \bm \beta \right\|_1 , {\text{ s.t. }} \ \lVert \bm z - \bm \Phi \bm \beta \rVert_2 \le \eta.
				\end{equation}	
				Problem \eqref{eq:l1bpdn} can be solved using the Lasso method \cite{Eldar2012}, which finds the solution to the $ \ell_1 $ regularized least squares as 
				\begin{equation}
				\label{eq:l1noisy}
				\hat{\bm \beta} = \mathop{\arg \min }_{\bm \beta} \lambda \left\| \bm \beta \right\|_1 + \frac{1}{2}\lVert \bm z - \bm \Phi \bm \beta \rVert_2^2,
				\end{equation}
				where $ \eta $ and $ \lambda $ are predefined parameters.}}
	}
	
	{\Revise{Having obtained the estimate  $\hat{\bm \beta}$ using \ac{cs} methods, we can use it to identify which of these estimated parameters correspond to a true target of interest.  Elements with significant amplitudes in $\hat{\bm \beta}$ are detected as dominant scattering points. Denote by $\widehat{\mySet{S}} $ the support set indexing these dominant scattering points, whose range-Doppler parameters are recovered from the corresponding indices in $\widehat{\mySet{S}}$. For example, we may use some threshold $ T_h $ to determine whether the amplitude is significant\cite[Page 295, Ch. 6]{richards2005fundamentals}. 
5			In this case, the support set is given by $\widehat{\mySet{S}}=\{s| \lvert [\bm \beta]_s \rvert> T_h \}$.  According to their Doppler estimates, these dominant scattering points are categorized into target of interests or clutter individually.}}
	
	{With the recovered range-Doppler values, \Revise{we can estimate} the angle and \Revise{refine the} scattering intensity,  as detailed in the following subsection. \Revise{The procedure for detecting which parameters correspond to targets of interest detailed above is based on $ \hat{\bm \beta} $. While this vector here represents the coarsely estimated scattering intensity, the procedure can be repeated for further separating targets from clutter based on the refined estimate obtained in the sequel.}}
	
	\vspace{-0.2cm}
	\subsection{Angle and Scattering Intensity Estimation}
	\label{subsec:angle}
	\vspace{-0.1cm}
	In this part, we refine the angle estimation of the scattering points, which are coarsely assumed within the transmit beam in the receive beamforming step, i.e., $\delta_{\vartheta} \approx 0$. 
	While the following formulation focuses on \ac{caesar}, the resulting algorithm is also applicable for \ac{wmar} as well as \ac{far}. 
	
	We estimate the directions of the scatterers  individually, as different  points may have different direction angles. Since after receive beamforming some directional information is lost in $\bm Z$, we recover the angles from the original data matrix $\bm Y_{\rm C}$ \eqref{eq:entry_y_sum}. 
	Using the obtained range and Doppler estimates, we first isolate echoes for each scattering point with an orthogonal projection, and apply a matched filter to estimate the direction angle of each scattering point. Finally, we use least squares to infer the scattering intensities. 
	
	%
	%
	\subsubsection{Echo isolation using orthogonal projection}
	In order to accurately estimate the angle of each scattering point, it is necessary to mitigate the interference between scattering points. To that aim, we use an orthogonal projection to isolate echoes from each \Revise{scatterer}.
	\par
	Let $\widehat{\mathcal{S}}$ be the support set of $\hat{\bm \beta}$,  and infer the normalized range and Doppler parameters $\{\tilde{r}_s, \tilde{v}_s\}$ from $\widehat{\mathcal{S}}$.
	According to \eqref{eq:entry_y_sum}, given these parameters, the original data vector from the $l$-th array element can be written as
	\begin{equation}
	\left[{\bm Y}_{\rm C}^T\right]_l = \left[{\bm \Psi}_l\right]_{\widehat{\mathcal{S}}} \left[{{\bm \gamma}}_l\right]_{\widehat{\mathcal{S}}} + \left[{\bm N}^T\right]_l,
	\end{equation}
	where ${\bm \Psi}_l \in \mathbb{C}^{N \times MN}$ has entries $\left[ {\bm \Psi}_l\right]_{n,s} := e^{j \tilde{r}_s c_{n,k}}e^{j \tilde{v}_s n \zeta_{n,k}}$, and ${{\bm \gamma}}_l \in \mathbb{C}^{MN}$ denotes the effective scattering intensities corresponding to all discrete range-Doppler grids, with $s$-th entry $ \left[{{\bm \gamma}}_l\right]_{s} := \tilde{\beta}_s \rho (n,k,\delta_{{{\vartheta}}_s}) e^{-j2\pi \Omega_{n,k} ld  \sin \vartheta_s /c }$. The intensities, containing unknown phase shifts and antenna gains due to angles $\vartheta_s$, are estimated as
	\begin{equation*}
	\left[{\hat{\bm \gamma}}_l\right]_{\widehat{\mathcal{S}}} =\mathop{\arg \min}\limits_{\left[{{\bm \gamma}}_l\right]_{\widehat{\mathcal{S}}}} \left\| \left[{\bm Y}_{\rm C}^T\right]_l - \left[{\bm \Psi}\right]_{\widehat{\mathcal{S}}} \left[{{\bm \gamma}}_l\right]_{\widehat{\mathcal{S}}} \right\|_2^2 
	= \left[{\bm \Psi}\right]_{\widehat{\mathcal{S}}}^{\dag} \left[{\bm Y}_{\rm C}^T\right]_l,
	\label{eqn:Scattering1}
	\end{equation*}
	where ${\bm A}^{\dag} = \left(\bm A^H \bm A \right)^{-1}{\bm A^H} $ and we assume that $\big| \widehat{\mySet{S}}\big|<N$ and $\bm A^H \bm A$ is invertible. The received radar echo from the $s$-th scattering point, ${\widehat {\bm Y}}_s \in \mathbb{C}^{L \times N}$, $s \in \widehat{\mathcal{S}}$, is then reconstructed by setting the $l$-th row as
	\begin{equation}
	\big[{\widehat {\bm Y}}_s^T\big]_l=  \left[ \bm \Psi \right]_s \left[{\hat{\bm \gamma}}_l \right]_s.
	\label{eqn:Isolate}
	\end{equation}
	%
	%
	\subsubsection{Angle estimation using matched filter}
	With the isolated echoes ${\widehat {\bm Y}}_s $ of the $s$-th scattering point, we  use a matched filter to refine the unknown angle $\vartheta_s$, which is coarsely assumed within the beam in the previous receive beamforming procedure, i.e., $\vartheta_s \in \varTheta:=\theta\!+\!\left[-\frac{\pi}{2L}, \frac{\pi}{2L}\right]$. 
	Using \eqref{eq:entry_y_sum}, we write the isolated echo as
	$ {\widehat {\bm Y}}_s  = \tilde{\beta}_s{ {\bm Y}}_s (\vartheta_s)  + \bm N_s$,
	where $\bm N_s$ denotes the noise matrix corresponding to the $s$-th scattering point. The entries of the steering matrix ${ {\bm Y}}_s (\vartheta_s)\in \mathbb{C}^{L \times N}$  are
	\begin{equation*}
	\left[{ {\bm Y}}_s (\vartheta_s) \right]_{l,n}\! :=\! \rho_{\rm C} (n,k,\delta_{{{\vartheta}}_s}) e^{j {{\tilde{r}}}_s c_{n,k}}e^{j{ {\tilde{v}}}_s n \zeta_{n,k}}e^{-j 2\pi \Omega_{n,k} ld {\sin {{\vartheta}}_s}/c},
	\end{equation*}
	which can be computed using \eqref{eq:rho} with given $\vartheta_s$ and the estimates of the range-Doppler parameters.
	Note that $\tilde{\beta}_s$, $\vartheta_s$ and $\bm N_s$ are unknown, and $\vartheta_s$ is of interest. The value of the intensity $\tilde{\beta}_s$ recovered next is refined in the sequel to improve accuracy. Here, we apply least squares estimation, i.e.,
	\vspace{-0.1cm}
	\begin{equation}
	\hat{ \vartheta}_s, \hat{\tilde{ \beta}}_s = \mathop{\arg \min} \limits_{{ \vartheta}_s, {\tilde{ \beta}}_s} \big\|{\rm vec}\big(\widehat{\bm Y}_s \big) - \tilde{\beta}_s {\rm vec}\left( {\bm Y}_s (\vartheta_s)  \right) \big\|_2^2.
	\label{eq:lsonys}
	\vspace{-0.1cm}
	\end{equation}
	Substituting $\hat{\tilde{ \beta}}_s = \left({\rm vec}\left( {\bm Y}_s (\vartheta_s)  \right)\right)^{\dag} {\rm vec}\big(\widehat{\bm Y}_s \big) = \frac{{\rm tr}\left( {\bm Y}_s^H (\vartheta_s){\widehat {\bm Y}}_s \right)}{\left\| {\bm Y}_s(\vartheta_s)\right\|_F^2}$ into \eqref{eq:lsonys} yields a matched filter 
	\vspace{-0.1cm}
	\begin{equation}
	\label{eq:matchedfilter}
	\hat{ \vartheta}_s = \mathop{\arg \max}\limits_{\vartheta_s \in \varTheta}\frac{ \big| {\rm tr}\big( {\bm Y}_s^H (\vartheta_s){\widehat {\bm Y}}_s \big)   \big|^2}{\left\| {\bm Y}_s(\vartheta_s)\right\|_F^2}.
	\vspace{-0.1cm}
	\end{equation}
	The angle $\hat{ \vartheta}_s$ is estimated for each $s \in \widehat{\mySet{S}}$ via \eqref{eq:matchedfilter} separately. 
	%
	%
	\subsubsection{Scattering intensity estimation using least squares}
	When $\delta_{\vartheta_s} \ne 0$, there exist approximation errors in \eqref{eq:signal_beamforming} and the resultant intensity estimate $\hat{\bm \beta}$. We thus propose to refine the estimation of ${\bm \beta}$ from the original data matrix $ \bm Y_{\rm C}$ once the range-Doppler and angle parameters are acquired.
	Given estimated angles $\hat{ \vartheta}_s$, we concatenate the steering vectors into
	\begin{equation*}
	\bm C := \Big[ {\rm vec}\big({ \bm Y}_{s_0}\big( \hat{ \vartheta}_{s_0}\big) \big),  {\rm vec}\big({ \bm Y}_{s_1}\big( \hat{ \vartheta}_{s_1}\big) \big), \dots \Big], 
	\end{equation*}
	$s_0, s_1,\dots \in \widehat{\mathcal{S}}$. The model \eqref{eq:entry_y_sum} is rewritten as $  {\rm vec}\left( \bm Y_{\rm C} \right) = \bm C \left[\bm \beta \right]_{\widehat{\mathcal{S}}} + \bm N$, and
	$\bm \beta$ can be re-estimated via least squares as
	\begin{equation}
	\big[{\hat {\bm \beta}}\big]_{\widehat{\mathcal{S}}}\!=\!\mathop{\arg \min} \limits_{\left[{\bm \beta}\right]_{\widehat{\mathcal{S}}}} \left\|{\rm vec}\left( \bm Y_{\rm C} \right) - \bm C \left[\bm \beta \right]_{\widehat{\mathcal{S}}} \right\|_2^2 
	\!= \! \bm C^{\dag} {\rm vec}\left({\bm Y }_{\rm C}\right) .
	\label{eqn:LeastSquares}
	\end{equation}
	The overall \Revise{parameter estimation} method is summarized as Algorithm~\ref{alg:Algo1}. 
	
	\begin{algorithm}
		\caption{\ac{caesar} target recovery}
		\label{alg:Algo1}
		\begin{algorithmic}[1]
			\STATE Input: Data matrix $\bm Y_{\rm C}$.
			\STATE \label{stp:MF0} Beamform $\bm Y_{\rm C}$ into $\bm Z$ via \eqref{eq:receive_beamforming_k}.
			\STATE \label{stp:MF1} Use \ac{cs} methods to recover \Revise{the indices of the dominant elements of $\bm \beta$, representing the parameters of the targets of interest,} denoted by $\widehat{\mySet{S}}$, from $\bm Z$ based on the sensing matrix $\bm \Phi$ \eqref{eq:phi}.
			\STATE \label{stp:MF2} Reconstruct the normalized range-Doppler parameters $\{ {\tilde r}_s, {\tilde v}_s\}$ from $\widehat{\mySet{S}}$ based on \eqref{eqn:BmatForm}.
			\STATE \label{stp:MF3} Isolate  $\bm Y_{\rm C}$ into multiple echoes $\{ \widehat{\bm Y}_s\}$ via \eqref{eqn:Isolate}. 
			\STATE \label{stp:MF4} Recover the angles $\{\vartheta_s\}$ from $\{ \widehat{\bm Y}_s\}$ via   \eqref{eq:matchedfilter}.
			\STATE \label{stp:MF5} Refine the scattering intensities $\{\beta_s\}$ using   \eqref{eqn:LeastSquares}. %
			\STATE Output: parameters $\{ {\tilde r}_s, {\tilde v}_s, \vartheta_s, {\tilde \beta}_s\}$.
		\end{algorithmic}
	\end{algorithm}
	
	\vspace{-0.2cm}
	\section{Comparison to Related Radar Schemes}
	\label{sec:comparison}
	\vspace{-0.1cm}
	We next compare our proposed \ac{wmar} and \ac{caesar} schemes, and discuss their relationship with relevant previously proposed radar methods.
	
	\vspace{-0.2cm}
	\subsection{Comparison of WMAR, CAESAR, and FAR}
	\label{subsec:Compare}	
	\vspace{-0.1cm}
	We compare our proposed \Revise{techniques} to each other, as well as to \ac{far}, which is a special case of both \ac{wmar} and \ac{caesar} obtained by setting $K=1$. We focus on the following aspects: 1) \Revise{instantaneous bandwidth}; 2) the number of measurements in a \ac{cpi}; and 3) \Revise{\ac{snr}}. A numerical comparison of the target recovery performance of the considered radar schemes is provided in Section \ref{sec:sim}.
	\par
	In terms of \Revise{instantaneous bandwidth}, recall that \ac{caesar} and \ac{far} use narrowband transceivers, and only a monotone signal is transmitted or received by each  element. In \ac{wmar}, $K$ multi-tone signals are sent and received simultaneously in each pulse, thus it requires instantaneous wideband components. 
	\par
	To compare the number of obtained measurements, we note that for each \ac{crc},  \ac{wmar} acquires a data cube with $NLK$ samples, while \ac{far} and \ac{caesar} collect $NL$ samples in the data matrix. After receive beamforming, the number of observations become $N$,  $NK$ and $NK$, for \ac{far}, \ac{caesar}, and \ac{wmar}, respectively, via \eqref{eq:receive_beamforming_k}. This indicates that the multi-carrier waveforms of \ac{caesar} and \ac{wmar} increase the number of measurements after receive beamforming.
	\par
	\Revise{{The aforementioned radar schemes also differ in their \ac{snr}, as the transmitted power and antenna gains differ. Here, as in \cite[Page 304, Ch. 6]{richards2005fundamentals}, \ac{snr} refers to the ratio of the power of the signal component to the power of the  noise component after coherently accumulating the radar returns. To see this difference, we consider the case when there exists a target with range-Doppler-angle $ (0,0,0) $, scattering coefficient $ \beta $ and the noise elements in $ \bm N $ are i.i.d. zero-mean proper-complex Gaussian with variance $ \sigma^2 $. In this case, coherent accumulation of radar returns reduces to  $ \sum_{l,n,k}{[\bm Y_{\rm W}]_{l,n,k}}  $ in \ac{wmar} and  $ \sum_{l,n}{[\bm Y_{\rm C}]_{l,n}}  $ in \ac{caesar} (\ac{far} can be regarded as a special case of \ac{wmar}/\ac{caesar} with $ K=1 $), and the antenna gains are $ \rho_{\rm W} = L $ and $ \rho_{\rm C} = L/K $, respectively. It follows from \eqref{eq:entry_y_sumW} that the signal amplitude in radar returns of \ac{wmar} is $ L |\beta|/ \sqrt{K} $. After coherent accumulation, the amplitude becomes $NLK \cdot L |\beta|/ \sqrt{K} $, leading to a signal power of $ N^2L^4K |\beta|^2 $, while the power of the noise component becomes $ NLK \cdot \sigma^2 $.  Hence, the \ac{snr} of \ac{wmar} is ${NL^3|\beta|^2}/{\sigma^2}$. In \ac{caesar}, the amplitude of the signal component in $ \bm Y_{\rm C} $ is $ L/K |\beta| $, which becomes $NL \cdot L/K |\beta| $ after accumulation. Since there are only $ NL $ noise elements in \ac{caesar}, the noise power after accumulation is  $ NL \cdot \sigma^2 $ and the resultant \ac{snr} is ${NL^3|\beta|^2}/(K^2\sigma^2)$. Letting $ K = 1 $ implies that the \ac{snr} of \ac{far} is also ${NL^3|\beta|^2}/\sigma^2$. The \ac{snr} calculation indicates that \ac{caesar} has an \ac{snr} loss by a factor of $ K^2 $ compared to \ac{wmar} and \ac{far}. This loss stems from the fact that \ac{caesar} uses a subset of the antenna array for each carrier, and thus has lower antenna gain than \ac{far} and \ac{wmar}. This \ac{snr} reduction can affect the performance of \ac{caesar} in the presence of noise, as numerically demonstrated in  Section~\ref{sec:sim}. 
		}
	}

	The above comparison reveals the tradeoff between \Revise{instantaneous bandwidth requirement},  number of observations, and \Revise{post-accumulation \ac{snr}}. 
	Among these three factors, the number of observations is crucial to the target recovery performance especially in complex electromagnetic environments, where some observations may be discarded due to strong interference \cite{Wang2006,Rao2011}. {The proposed multi-carrier schemes, \ac{wmar} and \ac{caesar}, are numerically shown to outperform \ac{far} in Section \ref{sec:sim}, despite the gain loss of \ac{caesar}. 
		\ac{caesar} also achieves performance within a relatively small gap compared to \ac{wmar}, while avoiding the usage of instantaneous wideband components.}
	The resulting tradeoff between number of beamformed observations and \Revise{\ac{snr}}, induced by the selection of $K$, is not the only aspect which must be accounted for when setting the value of $K$, as it also affects the frequency agility profile. In particular, smaller $K$ values result in increased spectral flexibility, as different pulses are more likely to use non-overlapping frequency sets. Consequently, in our numerical analysis in Section \ref{sec:sim} we use small values of $K$, for which the gain loss between \ac{caesar} and \ac{wmar} is less significant, and increased frequency agility is maintained. 
	In addition, \ac{caesar} can also exploit its spatial agility character, which is not present in \ac{far} or \ac{wmar}, to realize a DFRC system, as discussed in our companion paper \cite{Huang2019}.

	
	\vspace{-0.2cm}
	\subsection{Comparison to Previously Proposed Schemes}
	\label{subsec:Compare2}	
	\vspace{-0.1cm}
	Similarly to \ac{caesar}, previously proposed \ac{fdmamimo} radar \cite{Sun2014}, \ac{summer} \cite{Cohen2018}, and \ac{fda} radar \cite{Antonik2006, Liu2017} schemes also transmit a monotone waveform from each antenna element while different elements simultaneously transmit multiple carrier frequencies. 
	The main differences between our approaches and these previous methods are beam pattern \Revise{and} frequency agility. 
	Due to the transmission of diverse carrier frequencies from different array elements of \ac{fdmamimo}/\ac{summer}/\ac{fda}, the array antenna does not form a focused transmit beam and usually illuminates a large field-of-view \cite{Eli2013}. This results in a transmit gain loss which degrades the performance, especially for track mode, where a high-gain directional beam is preferred \cite{Eli2013}. 
	By transmitting each selected frequency with an antenna array (the full array in \ac{wmar} and a sub-array in \ac{caesar}), our methods achieve a focused beam pattern that facilitates accurate target recovery.
	
	Furthermore, \ac{fdmamimo} and \ac{fda} transmit all available frequencies simultaneously, and thus do not share the advantages of frequency agility, e.g., improved \ac{eccm} and \ac{emc} performance, as the multi-carrier version of \ac{summer} and the proposed \ac{wmar}/\ac{caesar}.
	In addition, \ac{fdmamimo}, \ac{summer} and \ac{wmar} receive instantaneous wideband signals with every single antenna, \Revise{as opposed to} \ac{fda} \cite{Liu2017} and \ac{caesar}, which use narrowband receivers.
	
	To summarize, \Revise{we compare in Table \ref{tbl:radars}  the main characteristics of these radar schemes.} {Unlike previously proposed radar methods, our proposed \Revise{techniques} are based on phased array antenna and frequency agile waveforms to achieve directional transmit beam and high resistance against interference. In terms of \Revise{instantaneous bandwidth},  \ac{caesar} is preferred for its usage of monotone waveforms and simple instantaneous narrowband receiver.}	
	\begin{table*}[t]
		\centering
		\caption{\Revise{Comparison between radar schemes}}
		\vspace{-0.2cm}
		\label{tbl:radars}
		\begin{tabular}{c|c|c|c|c|c}
			\hline
			Characters & Frequency agility & Beam pattern & \# of observation & Transmit bandwidth    & Receive bandwidth \\
			\hline\hline 
			\ac{wmar} & \textbf{Yes} & \textbf{Focused} & Moderate & Large & Large \\
			\hline
			\ac{caesar} & \textbf{Yes} & Moderately focused & Moderate & \textbf{Small} & \textbf{Small} \\
			\hline
			\ac{far} & \textbf{Yes} & \textbf{Focused} & Small & \textbf{Small} & \textbf{Small} \\
			\hline
			\ac{fdmamimo} & No & Omnidirectional & \textbf{Large} & \textbf{Small} & Large \\
			\hline
			\ac{summer} & \textbf{Yes} & Omnidirectional & Moderate & \textbf{Small} & Large \\
			\hline
			\ac{fda} & No & Omnidirectional & \textbf{Large} & \textbf{Small} & \textbf{Small} \\
			\hline
			
		\end{tabular}  
	\end{table*}

	%
	%
	\vspace{-0.2cm}
	\section{Performance Analysis of Range-Doppler Reconstruction}
	\label{sec:analysis}
	\vspace{-0.1cm}
	Range-Doppler reconstruction plays a crucial role in  target recovery. 
	This section presents a theoretical analysis of range-Doppler recovery \Revise{using \ac{cs} methods}.
	Since both \ac{wmar} and \ac{caesar} are generalizations of \ac{far}, the following analysis is inspired by the study of \ac{cs}-based \ac{far} recovery in \cite{Huang2018}. In particular, we extend the results of \cite{Huang2018}  to multi-carrier waveforms, as well as to extremely complex electromagnetic environments, where some transmitted pulses are interfered by intentional or unintentional interference. In the presence of such interference, only partial observations in the beamformed matrix $\bm Z$ remain effective for range-Doppler reconstruction.
	To present the analysis, we first briefly review some basic \Revise{analysis techniques} of \ac{cs} in Subsection \ref{subsec:CS}, followed by the range-Doppler recovery performance analysis  in Subsection~\ref{subsec:fullobserv}. 
	
	\vspace{-0.2cm}
	\subsection{\Revise{Preliminaries}}
	\label{subsec:CS}
	\vspace{-0.1cm}
	There have been extensive studies on theoretical conditions that guarantee unique recovery for noiseless models or robust recovery for noisy models \cite{Eldar2012, Eldar2015}. The majority of these studies characterize conditions and properties of the measurement matrix $\bm \Phi$, including spark, mutual incoherence property (MIP) and restricted isometry property (RIP).
	\par
	Following \cite{Huang2018}, we focus on the MIP. A sensing matrix ${\bm \Phi}$ is said to satisfy the MIP when its coherence, defined as
	\begin{equation}
	\label{eq:mu}
	\mu(\bm \Phi):=\max \limits_{i \neq j} \frac{\left|\left[\bm \Phi \right]_i^H\left[\bm \Phi \right]_j\right|}{\big\|\left[\bm \Phi \right]_i\big\|_2 \big\|\left[\bm \Phi \right]_j\big\|_2},
	\end{equation}
	is not larger than some predefined threshold.
	Bounded coherence ensures unique or robust recovery using a variety of computationally efficient \ac{cs} methods. We take $\ell_1$ optimization as an example to explain the bounds on matrix coherence. 
	In the absence of noise, the uniqueness of the solution to \eqref{eq:l1} is guaranteed by the following theorem:
	\begin{Thm}[\hspace{1sp}\cite{Fuchs2004}]{\label{thm:mip}}
		Suppose the sensing matrix ${\bm \Phi} $ has coherence $\mu(\bm \Phi)<\frac{1}{2S-1}$. If ${\bm \beta}$ solves \eqref{eq:l1} and has  support size at most $S$, then ${\bm \beta}$ is the unique solution to \eqref{eq:l1}.
	\end{Thm} 
	{
		\Revise{In noisy cases, the following result shows that $\mu(\bm \Phi)<\frac{1}{2S-1}$ also guarantees stable recovery.
			\begin{Thm}[\hspace{1sp}\cite{Cai2010}]{\label{thm:mipnoisy}}
				Consider the model \eqref{eq:noisySigModel} with $ \lVert \bm z \rVert_2 \le \epsilon \le \eta $ and $ \lVert \bm \beta \rVert_0 \le S $, and let sensing matrix ${\bm \Phi} $ have coherence $\mu(\bm \Phi)<\frac{1}{2S-1}$. If $\hat{\bm \beta}$ solves \eqref{eq:l1bpdn}, then 
				\begin{equation}
				\lVert \hat{\bm \beta} - \bm \beta \lVert_2 \le \frac{\sqrt{3(1+\mu)}}{1-(2S-1)\mu}(\eta + \epsilon).
				\end{equation}
		\end{Thm} }
	}
	
	Based on Theorems \ref{thm:mip} and \ref{thm:mipnoisy}, we next analyze the coherence measure of the sensing matrix $\bm \Phi$ in \eqref{eq:phi} for \ac{caesar} and \ac{wmar} (whose sensing matrices are identical), and establish the corresponding performance guarantees.

	\vspace{-0.2cm}
	\subsection{Performance Analysis}
	\label{subsec:fullobserv}
	\vspace{-0.1cm}
	Here, we analyze the range-Doppler reconstruction of \ac{wmar} and \ac{caesar}. Since the sensing matrix $\bm \Phi$ is random, we start by analyzing its statistics, and then derive conditions that ensure unique recovery by invoking Theorem~\ref{thm:mip}. 
	\par
	We assume that the frequency set $\mySet{F}_n$ is uniformly i.i.d. over $\left\{\mySet{X} \left| \mySet{X} \subset \mathcal{F}, |\mySet{X}| = K \right. \right\}$. For mathematical convenience, in our analysis we  adopt the narrow relative bandwidth assumption from \cite{Huang2018}, i.e., $\zeta_{n,k} \approx 1$, such that \eqref{eq:phi} becomes
	\begin{equation}\label{eq:phi_approx}
	\left[{\bm \Phi} \right]_{k+nK, l+mN} = e^{j\frac{2\pi m}{M}c_{n,k}+ j\frac{2\pi l}{N} n}.
	\end{equation}
	Numerical results in \cite{Huang2018} indicate that large relative bandwidth has negligible effect on the MIP of $\bm \Phi$. 
	In addition, recall that all the targets precisely lie on the predefined grid points, as assumed in Subsection \ref{subsec:RangeDoppler}. 
	{\Revise{Here, we adopt the on-the-grid assumption for mathematical convenience. Consequently, the accuracy and actual resolutions of range and Doppler reconstruction results are restricted by the grid intervals, i.e., $ 2 \pi/M $ and $ 2 \pi/N $, respectively. In practical scenarios, we may use denser grid points, as will be discussed by simulations in Subsection \ref{subsec:accuracy}. Denser grid enhances the obtainable accuracy/resolution and alleviates the performance loss when the real parameters are off the grid points. However, the density of grid points cannot go to infinity, because denser grid affects the incoherence property of the observation matrix while increasing the memory requirements and the computational burden. An alternative approach to overcome the need to specify a range-Doppler grid is to utilize  off-the-grid \ac{cs} methods, see \cite{Huang2012c,Tang2013}.}
	}
	\par
	In complex electromagnetic environments, some of the radar echoes may be corrupted due to jamming or interference. Heavily corrupted echoes are unwanted and should be removed before processing in order to avoid their influence on the estimation of target parameters \cite{Wang2006,Rao2011}. 
	In this case, the corrupted radar returns are identified, as such echoes typically have distinct characters, e.g., extremely large amplitudes. These interfered observations are regarded as missing, where we consider two kinds of missing patterns: 1) pulse selective, i.e., all observations in certain pulses are missing, which happens when the interference in these pulses is intense over all sub-bands; 2) observation selective, namely, only parts of the observations are missed when the corresponding pulse is interfered. We consider the first case, assuming that the radar receiver knows which pulses are corrupted, and leave the analysis under the second case for future investigation. In particular, we adopt the missing-or-not approach \cite{Pham2008}, in which each pulse in $\bm z$ has a probability of $1-u$, $0<u<1$, to be corrupted, and the missing-or-not status of the pulses are statistically independent of each other.
	
	After removing the corrupted returns, only part of the observations in the beamformed vector $\bm z$ \eqref{eq:noisySigModel} are used for range-Doppler recovery. Equivalently, corresponding rows in  $\bm \Phi$ can be regarded as missing, affecting the coherence of the matrix and thus the reconstruction performance. 
	Denote by $\Lambda \subset \mySet{N}$ the random set of available pulse indexes and by  $\Lambda_{\ast} :=\left\{nK+k\left| n \in \Lambda, k \in \mySet{K} \right. \right\}$ the corresponding index set of available observations. The signal model \eqref{eq:noisySigModel} is now 
	\vspace{-0.1cm}
	\begin{equation}
	\label{eq:partial_measurements}
	{\bm z}_{\ast} ={\bm \Phi}_{\ast}{\bm \beta}+{\bm n}_{\ast},
	\vspace{-0.1cm}
	\end{equation}
	where ${\bm z}_{\ast}:=\left[{\bm z}\right]_{\Lambda_{\ast}} $, ${\bm \Phi}_{\ast}:=\left[{\bm \Phi} ^T\right]_{\Lambda_{\ast}}^T $, and ${\bm n}_{\ast}:=\left[{\bm n}\right]_{\Lambda_{\ast}} $. 
	\par
	Consider the inner product of two  columns in $\bm \Phi_{\ast}$, denoted  $[\bm \Phi_{\ast}]_{l_1}$ and $[\bm \Phi_{\ast}]_{l_2}$,  corresponding to grid points $\left( \frac{2\pi m_1}{M}, \frac{2\pi n_1}{N} \right)$ and $\left( \frac{2\pi m_2}{M}, \frac{2\pi n_2}{N} \right)$, respectively, $l_1, l_2 \in 0, 1,\dots,MN-1$, $m_1,m_2 \in \mySet{M}$, $n_1,n_2 \in \mySet{N}$. While there are $M^2N^2$ different pairs of $(l_1,l_2)$, the magnitude of the inner product $\big| \left[{\bm \Phi}_{\ast}\right]_{l_1}^{H} \left[{\bm \Phi}_{\ast}\right]_{l_2} \big|$, which determines the coherence of ${\bm \Phi}_{\ast}$, takes at most $MN-1$ distinct random values. To see this,  note that
	\vspace{-0.1cm}
	\begin{align}
	\!\!\left[{\bm \Phi_{\ast}}\right]_{l_1}^{H} \left[{\bm \Phi_{\ast}}\right]_{l_2} \!&=\! \sum \limits_{n \in \Lambda}\! \sum \limits_{k = 0}^{K-1}\! e^{-j \frac{2\pi m_1}{M} c_{n,k}\! -\! j \frac{2\pi n_1}{N}n}
	e^{j \frac{2\pi m_2}{M} c_{n,k} \!+ \!j \frac{2\pi n_2}{N}n} \notag\\
	&= \sum \limits_{n \in \Lambda} \sum \limits_{k = 0}^{K-1} e^{-j2\pi \frac{ m_1-m_2}{M} c_{n,k} - j 2\pi \frac{ n_1-n_2}{N}n},
	\label{eq:innnerproductofphi} 
	\vspace{-0.1cm}
	\end{align}
	indicating that the inner product depends only on the difference of the grid points, i.e., $m_1 - m_2$ and $n_1 - n_2$, and not on the individual values of the column indices $l_1$ and $l_2$. 
	It follows from \eqref{eq:innnerproductofphi} that the MIP of ${\bm \Phi_{\ast}}$ can be written as 
	\vspace{-0.1cm}
	\begin{equation}
	\mu(\bm \Phi_{\ast})=\max \limits_{\substack{(\Delta_m,\Delta_n)\\ \neq (0,0)}} \frac{1}{|\Lambda| K}\sum \limits_{n \in \Lambda} \sum \limits_{k = 0}^{K-1} e^{-j\Delta_m  c_{n,k} }
	e^{- j \Delta_nn}.
	\label{eqn:MIP1}
	\vspace{-0.1cm}
	\end{equation}
	where $\Delta_m:= 2 \pi \frac{m_1-m_2}{M}$, $\Delta_n := 2 \pi \frac{n_1-n_2}{N}$ take values~in~the sets $\Delta_m\!\!\in\!\!  \{\pm \frac{2\pi m}{M}  \}_{m\in\mySet{M}} $ and $\Delta_n \in\{ \pm \frac{2\pi n}{N}\}_{n \in \mySet{N}}$, respectively.	
	
	Next, we define
	\vspace{-0.1cm}
	\begin{equation}
	\chi_n \left(\Delta_m, \Delta_n \right) := I_{\Lambda}(n) \frac{1}{K} \sum \limits_{k = 0}^{K-1} e^{-j \Delta_m c_{n,k} - j \Delta_n n},
	\label{eqn:Chindef}
	\vspace{-0.1cm}
	\end{equation}
	where the random variable $I_{\Lambda}(n)$ satisfies $I_{\Lambda}(n) = 1$ when $n \in \Lambda$ and 0 otherwise.
	In addition, let
	\vspace{-0.1cm}
	\begin{equation}
	\chi \left(\Delta_m, \Delta_n \right) := \sum \limits_{n = 0 }^{N-1} \chi_n\left(\Delta_m, \Delta_n \right).
	\vspace{-0.1cm}
	\label{eqn:Chidef}
	\end{equation}
	Some of the magnitudes $\left| \chi \left(\Delta_m, \Delta_n \right)  \right|$ are duplicated since 
	$\chi \left(\Delta_m, \Delta_n \right)  =   \chi \left(\Delta_m\pm 2\pi, \Delta_n \pm 2\pi \right)$ and
	$\chi \left(\Delta_m, \Delta_n \right)  =   \chi^* \left(-\Delta_m, -\Delta_n \right)$. 
	To eliminate the duplication and remove the trivial nonrandom value $\chi(0,0)$, we restrict the values of $\Delta_m$ and $\Delta_n$ to $\Delta_m\in \{\frac{2\pi m}{M}  \}_{m\in\mySet{M}} $ and $\Delta_n \in\{\frac{2\pi n}{N}\}_{n \in \mySet{N}}$, respectively, and define the set
	\begin{eqnarray}
	\Xi := \left\{\left(\Delta_m, \Delta_n \right)| (m,n)\in\mySet{M} \times \mySet{N} \backslash (0,0)\right\}, 
	\end{eqnarray}
	with cardinality $|\Xi | = MN-1$, such that each value of $\big| \left[{\bm \Phi}_{\ast}\right]_{l_1}^{H} \left[{\bm \Phi}_{\ast}\right]_{l_2} \big|$ (except the trivial case $l_1 = l_2$) corresponds to a single element of the set $\Xi$. We can now write \eqref{eqn:MIP1} as
	\begin{equation}
	\mu(\bm \Phi_{\ast})=\max \limits_{\left(\Delta_m, \Delta_n \right) \in \Xi}   \frac{1}{|\Lambda|} \left|\chi \left(\Delta_m, \Delta_n \right)\right|.
	\label{eqn:MIP2}
	\end{equation}

	The coherence in \eqref{eqn:MIP2} is a function of the dependent random variables $\chi$ and $|\Lambda|$. To bound $\mu$, we  derive bounds on $\chi$ and $|\Lambda|$, respectively. To this aim, we first characterize the statistical  moments of $\chi_n\left(\Delta_m, \Delta_n \right)$ for some fixed $\left(\Delta_m, \Delta_n \right) \in \Xi$, which we denote henceforth as $\chi_n$, in the following lemma:
	\begin{lemma}
		\label{lem:asymGaussian_miss}
		The sequence of random variables $\{\chi_n\}$ satisfies 
		\begin{equation}
		\label{eq:Echindeltam}
		{\rm E}\left[\chi_n \right]
		=\begin{cases}
		u e^{j\Delta_n n},\ \text{if\ }\Delta_m = 0 ,\\
		0,\ \text{otherwise},
		\end{cases}
		\end{equation}
		
		\begin{equation}
		\label{eq:sumvecchi2}
		\sum_{n = 0}^{N-1}{\rm D}\left[ {\chi}_n \right]
		=\begin{cases}
		u(1-u) N,\ \text{if\ } \Delta_m = 0,\\
		\frac{M-K}{(M-1)K}uN,\ \text{otherwise}.
		\end{cases}
		\end{equation}
		Furthermore, for each $n \in \mySet{N}$,
		\begin{equation}
		\label{eq:chibound2}
		\left|  {\chi}_n - {\rm E}[\chi_n] \right| \le 1, \quad {\text{w.p. }} 1.
		\end{equation}
		%
		%
	\end{lemma}
	\begin{proof}
		See Appendix \ref{app:Proof2}.
	\end{proof}
	
	Using Lemma \ref{lem:asymGaussian_miss}, the probability that the magnitude $|\chi|$ is bounded can be derived as in the following Corollary:
	\begin{corollary}\label{cor:Rayleigh}
		Let $V := \max \left\{u(1-u) N, \frac{M-K}{(M-1)K}uN\right\}$. For any $\left(\Delta_m, \Delta_n \right) \in \Xi$ and $\epsilon \le V$ it holds that
		\begin{equation}\label{eq:chiboundpr}
		\mathbb{P}\left( \left| \chi \right| \ge \sqrt{V} + \epsilon \right) \leq e^{-\frac{\epsilon^2}{4V}}.
		\end{equation}
	\end{corollary}
	\begin{proof}
		Based on the definition \eqref{eqn:Chindef} and the independence assumption on the frequency selection and missing-or-not status of each pulse, it holds that $\left\{\chi_n - {\rm E} \left[\chi_n\right]\right\}_{n \in \mySet{N}}$ are independent zero-mean complex-valued random variables. Now, since 
		\vspace{-0.1cm}
		\begin{equation}
		\sum_{n = 0}^{N-1} {\rm E} \left[\chi_n\right]
		=\begin{cases}
		u \sum_{n = 0}^{N-1}e^{j\Delta_n n},\ \text{if\ }\Delta_m = 0 ,\\
		0,\ \text{otherwise},
		\end{cases}
		\vspace{-0.1cm}
		\end{equation}
		according to \eqref{eq:Echindeltam} and $\sum_{n = 0}^{N-1}e^{j\Delta_n n} = \frac{1-e^{j\Delta_n N}}{1-e^{j\Delta_n}}$ equals 0 for $\Delta_n \in\{\frac{2\pi n}{N}\}_{n \in \mySet{N}\backslash \{0\}}$, recalling that $(\Delta_m, \Delta_n) \ne (0,0)$, we have $\sum_{n = 0}^{N-1} {\rm E} \left[\chi_n\right] = 0$. Then, it holds that $ \sum_{n = 0}^{N-1} \left( {\chi}_n - {\rm E} \left[\chi_n\right] \right)    = \sum_{n = 0}^{N-1}  {\chi}_n = \chi$. 
		Combining Bernstein's inequality  \cite[Thm. 12]{Gross2011} with the fact that by \eqref{eq:chibound2}, $ 	\left|  {\chi}_n - {\rm E}[\chi_n] \right| \le 1$, results in \eqref{eq:chiboundpr}. 
	\end{proof}
	We next derive a bound on the number of effective pulses $|\Lambda|$ in the following lemma:
	\begin{lemma}
		\label{lem:Lambdabound}
		For any $t>0$, it holds that
		\begin{equation}\label{eq:Lambdaboundpr}
		\mathbb{P}\left( \left| \Lambda \right| \le uN - t \right) < e^{-\frac{2t^2}{N}}.
		\end{equation}
	\end{lemma}
	\begin{proof} Since, by its definition, $ \left| \Lambda \right|$ obeys a binomial distribution, \eqref{eq:Lambdaboundpr} is a direct consequence of \cite[Thm. 1]{Okamoto1959}.
	\end{proof}
	
	Based on the requirement $\mu(\bm \Phi_{\ast})>\frac{1}{2K-1}$ in Theorem \ref{thm:mip}, we now use Corollary \ref{cor:Rayleigh} and Lemma \ref{lem:Lambdabound} to derive a sufficient condition on the radar parameters $M$, $N$, $K$, as well as the intensity of interference $1-u$, guaranteeing that the measurement matrix $\bm \Phi_{\ast}$ meets the requirement with high probability. This condition is stated in the following theorem:
	\begin{Thm}\label{thm:MIP}
		For any constant $\delta > 0$, the coherence of $\bm \Phi_{\ast}$ satisfies 
		$	\mathbb{P}\left( \mu({\bm \Phi_{\ast}}) \le \frac{1}{2S-1}  \right) \ge 1-\delta$   when
		\begin{equation}\label{eq:Sbound}
		S \leq \frac{uN/{\sqrt V}}{1\! + \!\sqrt{ 2\left(\log 2 |\Xi|\! -\! \log \delta \right)}}\frac{1\!+\!\frac{1}{{2\sqrt{2N}}u}}{2}\! -\! \sqrt{ \frac{N}{32V}}\! +\! \frac{1}{2}.
		\end{equation}
	\end{Thm}
	\begin{proof} 
		See Appendix \ref{app:Proof3}.
	\end{proof}
	\par
	
	Recall that the value of $V$ depends on the quantities $u$, $K$ and $M$. When $u$ is reasonably large such that $1-u \ge \frac{M-K}{(M-1)K}$, we have $V = \frac{M-K}{(M-1)K}uN$.
	When there is no noise in radar returns, a number of scattering points (on the grid) in the scale of $S = O\left( \sqrt{\frac{KuN}{\log MN}}\right)$ guarantees a unique reconstruction of range-Doppler parameters with high probability according to Theorems \ref{thm:mip} and \ref{thm:MIP}. Note that this rather simple asymptotic condition assumes that $\frac{M-1}{M-K} \approx 1$, i.e., that the overall number of available frequencies $M$ is substantially larger than the number of frequencies utilized in each pulse $K$, thus ensuring the agile character in frequency domain. 
	Compared to the asymptotic condition  $O\left( \sqrt{\frac{N}{\log MN}}\right)$ of \ac{far} with full observations \cite{Huang2018}, we find that the presence of corrupted observations, i.e. when $u<1$, leads to degraded  range-Doppler reconstruction performance. However, by increasing the number of transmitted frequencies in each pulse $K$ while maintaining $K \ll M$, the performance deterioration due to missing observations can be mitigated, enhancing the interference immunity of the radar in extreme electromagnetic environments. In the special case that $K=1$ and $u =1$, i.e., \ac{far} in an interference free environment, the two conditions coincide as $ O\left( \sqrt{\frac{KuN}{\log MN}}\right) = O\left( \sqrt{\frac{N}{\log MN}}\right) $.
	\par
	The above condition is proposed for the noiseless case, indicating that the inherent  target\Revise{/clutter} reconstruction capacity increases with $\sqrt{K}$. In practical noisy cases, the reconstruction performance does not monotonically increase with $K$, because the transmit power of each frequency decreases with $K$, thus degrading the \ac{snr} in both \ac{caesar} and \ac{wmar}. Particularly, in \ac{caesar}, larger $K$ means that less antennas ($L/K$) are allocated to each frequency, which affects the radiation beam and enlarges the gain loss. In addition, a small $K$ maintains the practical advantages of frequency agility in terms of, e.g., \ac{eccm} and \ac{emc} performance.
	

	\vspace{-0.2cm}
	\section{Simulation Results}
	\label{sec:sim}
	\vspace{-0.1cm}
	{In this section, we numerically compare the performance of \ac{wmar}, \ac{caesar},  and  \ac{far}  in noiseless\Revise{/noisy, clutter, and/or} jamming environments. The performance is evaluated in terms of target detection probability, accuracy and resolution, probability of correct reconstruction, and mutual interference, as presented in Subsections \ref{subsec:detect} - \ref{subsec:mi}, respectively.

		We consider a frequency band starting from $f_c = 9$ GHz, with $M=$ \Revise{4} available carriers and carrier spacing of $\Delta f = 1$ MHz. The radar system is equipped with an antenna array of $L = 10$ elements with  spacing of $d = \frac{c}{2f_c}$, and utilizes $N=32$ pulses focusing on  $\theta = 0$. \ac{caesar} and \ac{wmar} use $K=2$ frequencies at each pulse.
		\Revise{In noisy scenarios, we use the term \ac{snr} for the post-accumulation \ac{snr} of \ac{wmar}/\ac{far}, i.e., $ {NL^3|\beta|^2}/{\sigma^2} $ as derived in Subsection \ref{subsec:Compare}. To guarantee fair comparison, we use the same definition for all the radar schemes. }   	
		In the presence of jamming, \Revise{we test the pulse selective missing pattern.}
		In order to implement  target recovery via Algorithm \ref{alg:Algo1} for the three radar schemes, we use the convex optimization toolbox \cite{convexjl} to implement  basis pursuit  \eqref{eq:l1} in noiseless cases or the Lasso algorithm \Revise{\eqref{eq:l1noisy} with $ \lambda = 0.5 $} in noisy setups for range-Doppler reconstruction. 
	}

	\subsection{Target Detection in Clutter Environment}
	\label{subsec:detect}
	{\Revise{We compare the detection performance of the proposed \ac{wmar} and \ac{caesar} schemes with \ac{far} and conventional fixed frequency radars, which can be regarded as a special case of \ac{wmar}, \ac{caesar} and \ac{far}, with $ K = 1 $ and $ \Omega_{n,k} = f_n = f_c $.  To that aim, we evaluate the detection probabilities, $ P_d $, of a moving target under ground clutter environment. 
			The ground clutter is modeled  as radar returns from many static scattering points with velocity being zero. Thus, moving targets and ground clutter are distinguishable by observing their Doppler values.  We denote by $ \mathcal{C} $ the index set corresponding to zero Doppler, i.e., $ \mathcal{C} = \{ n +mN | n = 0, m \in \mathcal{M}\} $, and by $ \mathcal{C}^c $ its complementary set, i.e., $ \mathcal{C}^c = \{ n +mN | n \in \mathcal{N}\backslash \{ 0\}, m \in \mathcal{M}\} $. 
			The number of target scattering points is $ S_{ t} =1 $. The normalized range and Doppler parameter of the target is determined by randomly selecting the index $ i $ from  $ \mathcal{C}^c $. The target intensity $ |\beta_i|, i\in \mathcal{C}^c $ is selected to match the desired \ac{snr}. 
			The number of clutter scattering points is set as $S_{ c} =1000$, and the intensity of each scattering point is  unity, with a random phase uniformly distributed over $ [0, 2\pi) $. The normalized range parameters $ \tilde{r}_c $ of these scatterers are uniformly distributed over  $ [0, 2\pi) $, and are not assumed to lie on the predefined grid points. Since these scatterers have identical velocity, the superposition of their echoes, i.e., the clutter signal, can be represented by radar returns from the grid points indexed by $\mathcal{C}$. 
			Echoes from both target and clutter are treated as unknown, and are reconstructed simultaneously by \ac{cs}. Particularly, we apply Lasso \eqref{eq:l1noisy}  for range-Doppler reconstruction, yielding $ \hat{\bm \beta} $, where elements indexed in $\mathcal{C}$ are regarded as equivalent clutter intensities and the remaining elements are regarded as intensities of moving targets. 
			The $ P_d $ curves are plotted versus \ac{snr}, and we choose  $ \sigma^2 = 1$.
		} 
		
		\Revise{
			The simulations are carried out in two parts. In the first part, there are only clutter and additive noises in the radar returns, without returns from moving targets. The resulting radar measurements are used to determine the detection thresholds for all radar schemes, respectively, under a given probability of false alarm, denoted $P_{fa}$. In the second part, these thresholds are used to detect the existence of a target from the received echoes, which now include noise as well as  returns from both clutter and moving target.   
			The target detection procedure is based on the estimated parameters $ \lvert \hat{\bm \beta}\rvert $, as detailed in Subsection \ref{subsec:RangeDoppler}. 
			In the first part, we set $P_{fa}= 10^{-3}$ and perform $10^5$ Monte Carlo trials. While radar systems typically operate at lower values of $P_{fa}$, we use this value for computational reasons. It can be faithfully simulated under the given number of trials, and the selected value of $P_{fa}$ provides a characterization and understanding of the behavior of the considered radar schemes. A false alarm is proclaimed if any  nonzero-Doppler element in $ \lvert \hat{\bm \beta}\rvert $  exceeds certain threshold $ T_h $, i.e., $\max_{i \in \mathcal{C}^c} \lvert \hat{ \beta}_i\rvert > T_h $. 
			In the second part, we execute $200$ Monte Carlo trials.  A successful detection is proclaimed if the estimated intensity of the moving target is larger than the threshold, $\lvert \hat{ \beta}_i\rvert > T_h $, $ i \in \mathcal{C}^c $. 
			To evaluate the influence of missing observation caused by jamming, we set the survival rate $ u = 0.7 $ and compare the resulting probability of detection, $ P_d $ curves with those of the full observation cases in Fig. \ref{fig:pd}. Note that the detection thresholds for these $ P_d $ curves are calculated individually.
		}
	} 

	\begin{figure}	
		\centerline{\includegraphics[width=\myfiguresize in] {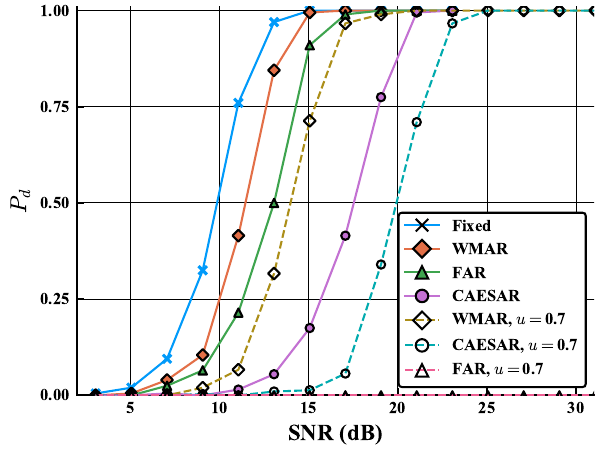}}
		\vspace{-0.2cm}
		\caption{\Revise{Detection probabilities $ P_d $ versus \ac{snr}. The label ``Fixed" represents the fixed frequency radar, and the $ P_d $  of \ac{far} with $ u=0.7 $ are zeros in the tested scenarios. }}
		\label{fig:pd}
	\end{figure} 
	
	{\Revise{As shown in Fig. \ref{fig:pd}, the fixed frequency radar has higher detection probabilities than the counterparts of frequency agile schemes. The advantage of fixed frequency radar stems from the property of its observation matrix $ \bm \Phi \in \mathbb{C}^{KN \times MN} $, where $ K = M = 1 $ and $ \bm \Phi $ becomes an orthogonal matrix, benefiting the Doppler reconstruction performance of \ac{cs} methods. While in the frequency agile schemes, generally it holds that $ M>K $, resulting in an incomplete observation matrix and degradation of clutter/target recovery performance. However, we  note that the fixed frequency scheme is vulnerable in a jamming environment. 
			In the full observation cases, \ac{wmar} outperforms \ac{far} because of the increased number of observations. Though \ac{caesar} has identical number of observations with \ac{wmar} after receive beamforming, the $ P_d $ values of \ac{caesar} are less than those achieved by \ac{wmar} with a \ac{snr} gap of approximately 6 dB. This follows since \ac{caesar} suffers from an antenna gain loss of $ K^2 $ as discussed in Subsection \ref{subsec:Compare} (i.e., 6 dB since $ K =2 $). When some of the observations are missing due to jamming, the detection probabilities are affected. Both  \ac{wmar} and \ac{caesar} suffer from an \ac{snr} loss of approximately 3 dB, while \ac{far} fails to detect any moving target in the scenarios under test. When \ac{far} is lacking in radar observations, the mutual coherence property of its observation matrix $ \bm \Phi_{\ast} $ becomes degraded, leading to many spurious peaks of high intensities in the recovery results $ \lvert \hat{ \bm \beta} \rvert $. These spurious peaks significantly increase the detection threshold, thus operating under a fixed $ P_{fa} $ of $ 10^{-3} $ results in notably reduced detection probability $ P_d $.}
		 
			\Revise{To summarize, we find from the simulation results that 1) the proposed multi-tones schemes (\ac{wmar} and \ac{caesar}) enhance the immunity against missing data over the single-tone \ac{far}, and 2) \ac{wmar} outperforms \ac{caesar} due to its higher antenna gain, which comes at the cost of increased  instantaneous bandwidth.}
	}
	
	\subsection{Accuracy and Resolution}
	\label{subsec:accuracy}
	{\Revise{Here, we compare the range, Doppler and angle estimate results under different range-Doppler grid points. To this aim, two sets of range-Doppler grid points are tested: one uses the standard grid as mentioned in Subsection \ref{subsec:RangeDoppler}, where the intervals of consecutive range and Doppler grid points are $ \Delta_{\tilde{r}} = \frac{2 \pi}{M} $ and $ \Delta_{\tilde{v}} = \frac{2 \pi}{N} $, respectively; The latter uses a denser grid, setting $ \Delta_{\tilde{r}} = \frac{2 \pi}{2M} $ and $ \Delta_{\tilde{v}} = \frac{2 \pi}{2N} $, and the consequent simulation results are denoted with label ``-2", e.g., \ac{caesar}-2. The number of scattering points is $ S = 1 $ without clutter. The normalized range-Doppler parameter, $ (\tilde{r}, \tilde{v} )$, of the scattering point is uniformly, randomly set over $ [0,2\pi)^2 $, and the angle is randomly set within the beam $ \vartheta \in \varTheta $. Under this setting, the ground truth of the range-Doppler parameter may be off the grid, which leads to inevitable estimation error. The \ac{cs} method  applied for range-Doppler reconstruction is based on  \eqref{eq:l1noisy}, and we estimate range-Doppler, denoted by $  (\hat{\tilde{r}},\hat{ \tilde{v}} )$, from the index of the element with maximum magnitude in $ \hat{ \bm \beta} $. We then use \ac{rmse} as the metric of accuracy, defined by $\sqrt{{\rm E}[ (\tilde{r} - \hat{\tilde{r}})^2 ]}$, taking normalized range as an example. The remaining settings are the same as those used in Subsection \ref{subsec:detect}. We run 500 Monte Carlo trials and the range, Doppler and angle accuracy results are shown in Figs. 
			\ref{fig:range_accuracy}, \ref{fig:velocity_accuracy} and \ref{fig:angle_accuracy}, respectively. 
		}
		
		\Revise{
			As expected, the \ac{rmse}s become lower when we increase \ac{snr}, while we observe in Figs. \ref{fig:range_accuracy} and \ref{fig:velocity_accuracy} that the \ac{rmse}s of range and Doppler estimates reach error floors as the \ac{snr} increases. The error floors depend on the grid intervals  $ \Delta_{\tilde{r}}  $ or $ \Delta_{\tilde{v}}  $, and denser grid points lead to lower error floors. The results also reveal that \ac{wmar} and \ac{far} have similar accuracy performance, while \ac{caesar} has an \ac{snr} loss of 6 dB in moderate \ac{snr} levels because of its lower antenna gain. In high \ac{snr} scenarios, the \ac{rmse}s of \ac{caesar} also reach the error floor.
		} 
	} 
	\begin{figure}	
		\centerline{\includegraphics[width=\myfiguresize in] {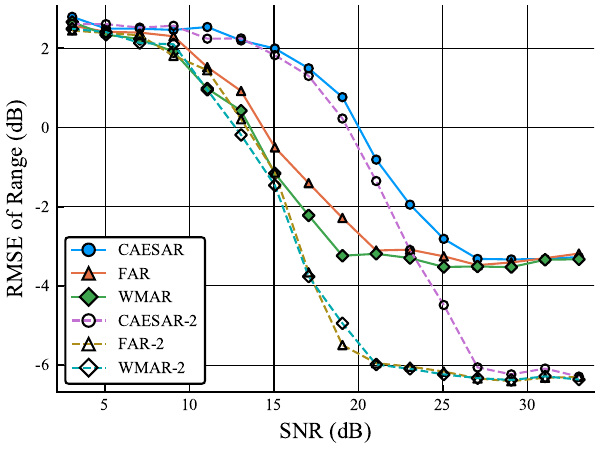}}
		\vspace{-0.3cm}
		\caption{\Revise{Range accuracy of range-Doppler reconstruction results.}}
		\label{fig:range_accuracy}
	\end{figure} 	
	
	\begin{figure}	
		\centerline{\includegraphics[width=\myfiguresize in] {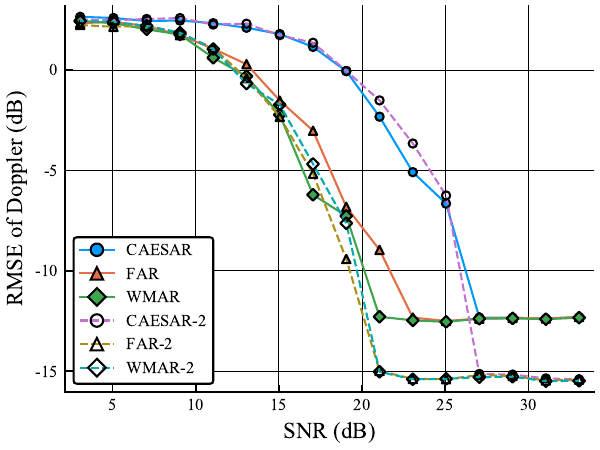}}
		\vspace{-0.2cm}
		\caption{\Revise{Doppler accuracy of range-Doppler reconstruction results.} }
		\label{fig:velocity_accuracy}
	\end{figure} 
	
	\begin{figure}	
		\centerline{\includegraphics[width=\myfiguresize in] {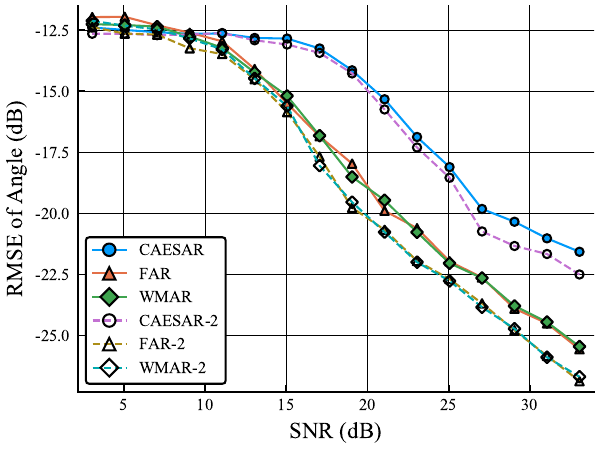}}
		\vspace{-0.2cm}
		\caption{\Revise{Angle accuracy versus \ac{snr}.}}
		\label{fig:angle_accuracy}
	\end{figure} 
	
	{\Revise{We next examine the ability of \ac{caesar}, \ac{far} and \ac{wmar} in separating closely spaced scattering points, i.e., obtainable resolution. In the simulations, we use dense grid points with intervals $ \Delta_{\tilde{r}} = \frac{2 \pi}{5M} $ and $ \Delta_{\tilde{v}} = \frac{2 \pi}{5N} $. We consider two closely spaced scattering points, of which angles are set $ \vartheta = 0 $ and range-Doppler parameters are on the grid. Particularly, we fix the range-Doppler parameter of one scattering point $ (\tilde{r}_1,\tilde{v}_1) = (0,0) $ and change the counterpart of the other scattering from $[\Delta_{\tilde{r}}, \frac{2 \pi}{M}] \times  [\Delta_{\tilde{v}}, \frac{2 \pi}{N}] $, such that the range/Doppler separation ($ \Delta_{\mathcal{R}} $/$ \Delta_{\mathcal{V}} $) between two points are changed. In this experiment, we disregard noise and the scattering intensities are both $ |\beta_1| = |\beta_2| = 1 $ with random phase. We use \eqref{eq:l1} for range-Doppler recovery, and the two most dominant elements in the estimate $ \hat{ \bm \beta} $ are regarded as scattering points. The indices of these two elements are compared with the corresponding ground truth, and a successful recovery (also referred to a hit) is proclaimed if both indices are correct. The number of Monte Carlo trails are 500. The achievable hit rate results versus separation between scattering points are shown in Fig. \ref{fig:resolution}. The results demonstrate that all the frequency agile schemes, \ac{caesar}, \ac{far} and \ac{wmar} have close performance in resolution, while the hit rates of \ac{caesar} and \ac{wmar} are slightly higher than those of \ac{far}, because they have more observations and the measurement matrix $ \bm \Phi $ has better coherence property.}}
	\begin{figure}	
		\centerline{\includegraphics[width=\myfiguresize in] {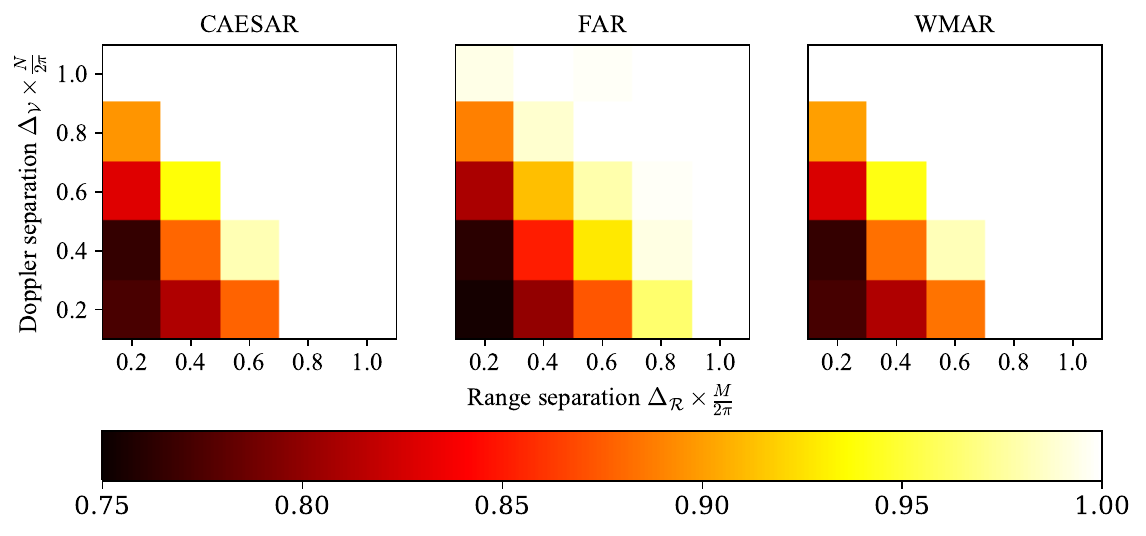}}
		\vspace{-0.2cm}
		\caption{\Revise{Hit rates of separating closely spaced scattering points.}}
		\label{fig:resolution}
	\end{figure}

	\vspace{-0.2cm}
	\subsection{Reconstruction of Multiple Scattering Points}
	\label{subsec:hitrate}
	\vspace{-0.1cm}
	{\Revise{In this subsection, we evaluate the proposed radar schemes in recovering a set of $ S $ scattering points in noiseless and noisy setups. In both scenarios, we use the standard grid points with grid intervals $ (\Delta_{\tilde{r}},\Delta_{\tilde{v}}) = (2\pi/M,2\pi/N) $. The range-Doppler parameters of scattering points are randomly selected from the grid points, angle parameters are randomly set from the continuous set $ \varTheta $, and scattering intensities are all set to unity. We apply \ac{cs} methods for range-Doppler recovery, and the indices of  $ S $ most significant entries in $ \hat{ \bm \beta} $ are regarded as elements of the estimated support set. 
			Hit rates are applied as performance metric, and a hit is proclaimed if the obtained support set is identical to the ground truth, which means all the range-Doppler parameters are reconstructed correctly.} 
		
		In the \Revise{noiseless} experiment, we simulate different numbers of recoverable scattering points, $S$. 
		\Revise{We set the  survival rate $u =0.4$ for jamming environments.}  
		The resulting hit rates versus $S \in \{1,\dots,N-1\}$ are depicted in Fig.~\ref{fig:full_noiseless_hitrate}. As expected, the hit rates decrease as $S$ increases. The performance of \ac{caesar} is within a very small gap of that achievable using \ac{wmar}, because \ac{caesar} and \ac{wmar} use the same amount of transmitted frequencies $K$, and the number of beamformed measurements is also the same. Hit rates of \ac{caesar} and \ac{wmar} exceed that of \ac{far} significantly. This gain stems from the fact that transmitting multi-carriers in each pulse of \ac{caesar} and \ac{wmar} increases the number of observations, and thus raises the number of recoverable scattering points.}
	\begin{figure}
		\centering
		\includegraphics[width=\myfiguresize in]{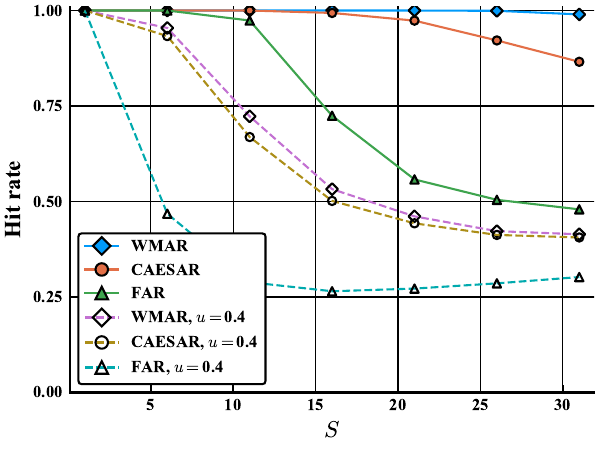}
		\vspace{-0.3cm}
		\caption{Range-Doppler recovery  versus $S$, noiseless setting.}
		\vspace{-0.4cm}
		\label{fig:full_noiseless_hitrate}
	\end{figure}
	\par
	{\Revise{We then consider the noisy case, and compare the range-Doppler recovery performance versus \ac{snr}, which is changed by varying $\sigma^2$. We set $S = 10$, and we let $ u = 0.4 $ for the jamming environment}. The hit rates of the range-Doppler parameters are depicted in Fig.~\ref{fig:noisy_hitrate}.}
	\begin{figure}
		\centering
		\includegraphics[width=\myfiguresize in]{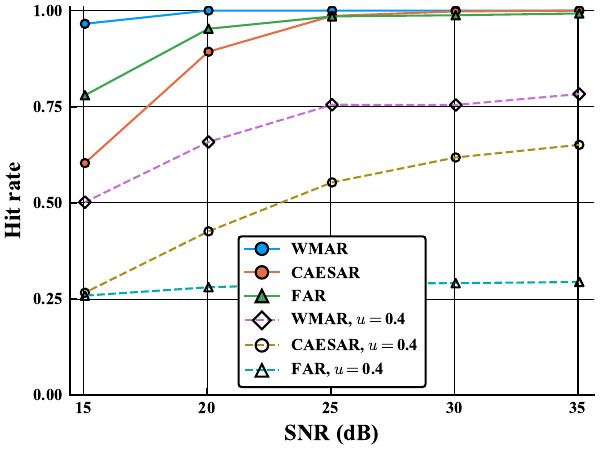}
		\vspace{-0.3cm}
		\caption{Range-Doppler recovery  versus \ac{snr}.}
		\vspace{-0.4cm}
		\label{fig:noisy_hitrate}
	\end{figure}
	
	{Observing Fig.~\ref{fig:noisy_hitrate}, we note that, as expected, \ac{wmar} achieves the best performance in range-Doppler reconstruction. While \ac{wmar} and \ac{caesar} have the same number of observations, \ac{caesar} has a lower antenna gain as noted in Subsection \ref{subsec:Compare}, which results in its degraded performance compared to \ac{wmar}. In \Revise{the full observation case with high \ac{snr}s, i.e., \ac{snr} $ \ge 25$ dB}, \ac{caesar} has higher hit rates than \ac{far} due to the advantage of increased number of transmitted frequencies, \Revise{while in low \ac{snr}s of less than 20 dB, }  \ac{far} exceeds \ac{caesar} owing to its higher antenna gain. 
		In the \Revise{jamming} scenario,  \ac{caesar} \Revise{outperforms \ac{far}}, and that \ac{far} almost fails to reconstruct scattering points (with hit rates \Revise{around 0.25}). The superiority of \ac{wmar}/\ac{caesar} over \ac{far} demonstrates the advantage of the proposed multi-carrier waveforms. 
		\par
		From the experimental results in \Revise{Subsections \ref{subsec:detect} - \ref{subsec:hitrate}}, we find that the multi-carrier signals used by \ac{caesar} and \ac{wmar} significantly enhance range-Doppler reconstruction performance over the monotone waveform in \ac{far}. The advantage becomes more distinct in jamming environments, where some radar measurements are invalid. In reasonably high \ac{snr} scenarios, the \Revise{reconstruction performance} of \ac{caesar}, which uses \Revise{narrowband} constant modulus waveforms \Revise{for each antenna element}, approach those of \ac{wmar}, which uses \Revise{instantaneously} wideband waveforms.} 
	
	\subsection{\Revise{Mutual Interference}}
\label{subsec:mi}
{\Revise{One of the main advantages of frequency agile transmission is its relatively low level of mutual interference, which implies that multiple transmitters can coexist in dense environments. To demonstrate this property of the proposed radar schemes, which all utilize some level of frequency agility, we next evaluate the unintended mutual interference of closely placed radars transmitting the same waveform pattern. We compare the frequency agile schemes of \ac{far}, \ac{wmar} and \ac{caesar}, with an instantaneous wideband radar which transmits all subbands simultaneously.
		In the simulation, we consider a scenario with $6$ radars operating independently. Mutual interference occurs if a reference radar is receiving echoes while another radar is transmitting  at the same subcarriers with their antenna beams directed towards each other. In this case the echoes of the reference radar at the conflicted subcarriers are corrupted. 
		The level of mutual interference is measured by the average number of uncorrupted subcarriers, denoted $K_{\rm u}$.}
	
	\Revise{We use $\mathbb{P}_{\rm Int}$ to represent the probability that one radar may interfere the reference radar, i.e., that it is radiating during the reception period of the reference radar and their beams are directed towards each other. The number of subcarriers transmitted in each pulse varies from $K = 1$ to $K = 4$, where $K = 1$ represents the \ac{far} while $K = 4$ represents the radar using full bandwidth. As \ac{wmar} and \ac{caesar} transmit the same number of subcarriers for a specific $K$, their performance on mutual interference is the same. To demonstrate the mutual interference intensity versus interference probabilities, we simulate $10^{6}$ Monte Carlo trials for each interference probability and calculate the average number of uncorrupted subcarriers as shown in Fig.~\ref{fig:MutualInterference}. 
		From the results, we observe that, as expected, when the interference probability  is small, e.g., less than $0.2$, radar systems transmitting more subcarriers are capable of effectively utilizing their bandwidth reliably. However, as the probability of interference grows, wideband radar induce severe mutual interference, resulting in a negligible average number of uncorrupted subcarriers for $\mathbb{P}_{\rm Int} > 0.6$. The frequency agile schemes, such as \ac{far} and \ac{wmar}/\ac{caesar} operating with $K=2$, are still capable of reliably utilizing a notable portion of their bandwidth in the presence of such high interference. These results  indicate that the less frequency agile the scheme is, the severer the mutual interference becomes in the  high interference probability regime.}}

\begin{figure}
	\centerline{\includegraphics[width= \myfiguresize in ]{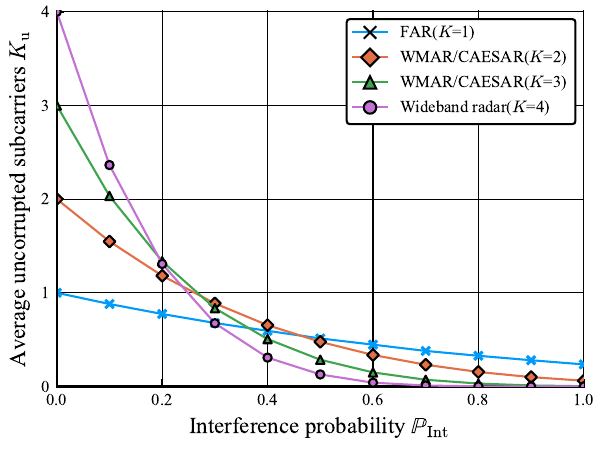}}
	\caption{\Revise{The average number of uncorrupted subcarriers $K_{\rm u}$ versus $\mathbb{P}_{\rm Int}$ for radar schemes with different number of transmit subcarriers, which varies from $K=1$ to $K = 4$.} } 
	\label{fig:MutualInterference}
\end{figure}
	
	\vspace{-0.2cm}
	\section{Conclusion}
	\label{sec:conclusion}
	\vspace{-0.1cm}
	In this work we developed two multi-carrier frequency agile schemes for phase array radars: \ac{wmar}, which uses wideband waveform\Revise{s}; and \ac{caesar}, which transmits monontone signals and introduces spatial agility. 
	We modeled  the received radar signal, and proposed an algorithm for target recovery. We then 
	characterized theoretical recovery guarantees. Our numerical results demonstrate that our proposed schemes achieve enhanced survivability in extreme electromagnetic environments. Furthermore, it is shown that \ac{caesar} is capable of achieving performance which approaches that of wideband radar, while utilizing narrowband transceivers. An additional benefit which follows from the introduction of frequency and spatial agility is the natural implementation of \ac{caesar} as a \ac{dfrc} system,  studied in a companion paper.
	
	\vspace{-0.2cm}
	\begin{appendix}

		\numberwithin{equation}{subsection}	
		%
		\vspace{-0.2cm}
		\subsection{Proof of Lemma \ref{lem:BeamPattern}}
		\label{app:BeamPattern}
		\vspace{-0.1cm} 
		In the following we prove  \eqref{eq:signal_beamforming} for \ac{caesar}. The proof for \ac{wmar} follows similar arguments and is omitted for brevity.
		
		Substituting the definitions of $\bm w, \bm P$ and $\bm Y$ into \eqref{eq:receive_beamforming_k} yields
		\begin{eqnarray}
		Z_{k,n} &=& \sum \limits_{l = 0}^{L-1} w_l\left( \theta, \Omega_{n,k}\right) \left[\bm p(n,k)\right]_l \sum \limits_{s = 0}^{S-1}\tilde{\beta}_s e^{j \tilde{r}_s c_{n,k}}\notag \\
		&&\quad e^{j \tilde{v}_s n \zeta_{n,k}}e^{-j 2\pi \Omega_{n,k} ld {\sin \vartheta_s}/c}\rho_{\rm C} (n,k,\delta_{\vartheta_s})\notag \\
		&=& \sum \limits_{s = 0}^{S-1}\sum \limits_{l = 0}^{L-1} \left[\bm p(n,k)\right]_l \tilde{\beta}_s e^{j \tilde{r}_s c_{n,k}}e^{j \tilde{v}_s n \zeta_{n,k}}\notag \\
		&&\quad e^{-j 2\pi \Omega_{n,k} ld (\sin \vartheta_s-\sin \theta)/c}\rho_{\rm C} (n,k,\delta_{\vartheta_s})\notag \\
		&=& \sum \limits_{s = 0}^{S-1}\tilde{\beta}_s e^{j \tilde{r}_s c_{n,k}}e^{j \tilde{v}_s n \zeta_{n,k}}\rho_{\rm C}^2 (n,k,\delta_{\vartheta_s}).
		\label{eq:receive_beamforming_k2}
		\end{eqnarray}
		Recall that when $\delta_{\vartheta_s} \approx 0$ it holds that $\rho_{\rm C} (n,k,\delta_{\vartheta_s}) \approx L/K = \sqrt{\Gain}$. 
		Then, \eqref{eq:receive_beamforming_k2} reduces to \eqref{eq:signal_beamforming}, proving the lemma. \qed
		
		\vspace{-0.2cm}
		\subsection{Proof of Lemma \ref{lem:asymGaussian_miss}}
		\label{app:Proof2}
		\vspace{-0.1cm}
		We first prove \eqref{eq:Echindeltam} and \eqref{eq:sumvecchi2}, after which we address \eqref{eq:chibound2}. 
		
		\subsubsection{Proof of \eqref{eq:Echindeltam} and \eqref{eq:sumvecchi2}}
		For brevity, let $p = - \Delta_m$, $q = - \Delta_n$, $I_n = I_{\Lambda}(n)$, and $B = \binom{M}{K}$. We set $\binom{M}{K} = 0$ when $M\leq 0$ or $K<0$, and $\binom{M}{0} = 1$ when $M> 0$.
		
		\par
		We first compute ${\rm E}\left[\chi_n \right]
		= \frac{1}{K}{\rm E}\big[ I_n e^{ jqn}   \sum _{k = 0}^{K-1} e^{jpc_{n,k}}\big]$. 
		The expectation is taken over the indicator $I_n$ and frequency codes $c_{n,k}$.
		Since they are independent and  ${\rm E}\left[ I_{n} \right] = u $, it holds that
		\begin{equation} 
		{\rm E}\left[\chi_n \right]
		=\frac{ue^{ jqn} }{K}{\rm E}\left[\sum \limits_{k = 0}^{K-1} e^{jpc_{n,k}}\right].
		\label{eq:miss:expectation}
		\end{equation} 
		Since $K$ frequencies are selected uniformly (but not independently), it follows that
		\begin{equation}
		\label{eq:expectation_sum}
		{\rm E}\left[ \sum \limits_{k = 0}^{K-1} e^{jpc_{n,k}}\right] = \frac{1}{B}\sum \limits_{i=0}^{B-1} \sum \limits_{k = 0}^{K-1} e^{jpm_{i,k} },
		\end{equation}
		where $m_{i,k}$ denotes the $k$-th frequency in the $i$-th combination. Out of these $B$ combinations, there are $\binom{M-1}{K-1}$ that contain a given selection $m \in \mySet{M}$. Thus, we have that
		\begin{align} 
		\!\!\sum_{i=0}^{B-1} \!\sum \limits_{k = 0}^{K-1}\! e^{jpm_{i,k} }\!=\!\binom{M\!-\!1}{K\!-\!1}\!\sum \limits_{m \!= \!0}^{M\!-\!1}\! e^{jpm} 
		%
		\!=\!
		\frac{BK}{M}\!\sum \limits_{m \!=\! 0}^{M\!-\!1}\! e^{jpm}. 
		\label{eqn:Proof1a} 
		\end{align}
		Substituting \eqref{eqn:Proof1a} into \eqref{eq:expectation_sum} yields
		\begin{equation}
		\label{eq:expectation_sum2}
		\begin{split}
		{\rm E}\left[ \sum \limits_{k = 0}^{K-1} e^{jpc_{n,k}}\right] = \frac{K}{M}\sum \limits_{m = 0}^{M-1} e^{jpm}=\frac{K}{M} \frac{1-e^{jpM}}{1-e^{jp}}.
		\end{split}
		\end{equation}
		As $p\! \in \left\{\frac{2\pi m}{M}\right\}_{m \in \mySet{M}}$, it holds that ${\rm E}\big[\sum_{k = 0}^{K-1} e^{jpc_{n,k}}\big] = K$ if $p=0$ and zero otherwise. 
		%
		Substituting this into \eqref{eq:miss:expectation}, we have
		\begin{equation}
		\label{eq:Echin}
		{\rm E}\left[\chi_n \right]
		=\begin{cases}
		u e^{jqn},\ \text{if\ }p = 0 ,\\
		0,\ \text{otherwise},
		\end{cases}
		\end{equation}
		which proves \eqref{eq:Echindeltam}.
		\par 
		To obtain $\!{\rm D}\left[ \chi_n\right]\!:=\!{\rm E}\big[\! \left| \chi_n \!-\! {\rm E}\! \left[ \chi_n \right] \right|^2 \big]$, we consider two cases, $p =0$ and $p \ne 0$.
		When $p =0$, we have $\chi_n = I_n e^{jqn}$ and 
		\begin{align} 
		&{\rm E}\left[ \left|
		I_{n}   e^{jqn} - u e^{jqn}\right|^2 \right]
		={\rm E}\left[ \left(
		I_{n}  - u \right)^2 \right]\notag \\
		&\qquad\stackrel{(a)}{=}  {\rm E}\left[  I_{n} + u^2 - 2 I_{n} u  \right]
		\stackrel{(b)}{=}  u - u^2, 
		\label{eq:abs2:p0}
		\end{align}
		where $(a)$ holds since $I_n^2 = I_n$ and in $(b)$ we apply ${\rm E}\left[ I_n\right] = u$. 
		\par
		When $p \ne 0$, the random variable $\chi_n$ has zero mean, and its variance is given by
		\begin{align}
		{\rm E}\left[\left| \chi_n\right|^2\right] &=  {\rm E}\left[     \frac{I_{n}}{K} \left| \sum \limits_{k = 0}^{K-1} e^{jpc_{n,k} + jqn} \right|^2 \right] \notag  \\
		&=  \frac{1}{K^2} {\rm E}\left[  I_{n }\cdot  
		\sum \limits_{k = 0}^{K-1} \sum \limits_{k' = 0}^{K-1}e^{jp\left(c_{n,k} - c_{n, k'}\right)} \right]\notag \\
		&=\frac{u}{K^2} {\rm E}\left[  
		\sum \limits_{k = 0}^{K-1} \sum \limits_{k' = 0}^{K-1}e^{jp\left(c_{n,k} - c_{n, k'}\right)} \right], 
		\label{eq:miss:abs2:start} 
		\end{align}
		where we use $I_n^2 = I_n$. 
		To compute \eqref{eq:miss:abs2:start},  we note that
		\begin{align}
		&{\rm E}  \left[\sum \limits_{k = 0}^{K\! - \!1}\sum \limits_{k' = 0}^{K\! - \!1}e^{jpc_{n,k}\! - \!jpc_{n,k'}} \right] \notag \\
		&=\frac{\binom{M\! - \!1}{K\! - \!1}}{B} \sum \limits_{m = 0}^{M\! - \!1} e^{jp\cdot0}   + \frac{\binom{M\! - \!2}{K\! - \!2}}{B}\sum \limits_{m = 0}^{M\! - \!1}\sum \limits_{\substack{m' = 0,\\ m'\neq m}}^{M\! - \!1} e^{jp(m\! - \!m')} \notag \\
		&\stackrel{(a)}{=}\frac{\binom{M\! - \!1}{K\! - \!1}\! - \!\binom{M\! - \!2}{K\! - \!2}}{B} M
		+ \frac{\binom{M\! - \!2}{K\! - \!2}}{B} \sum \limits_{m = 0}^{M\! - \!1}\sum \limits_{m' = 0}^{M-1} e^{jp(m-m')},
		\label{eq:abs2:Eabsx2_3}  
		\end{align}
		where $(a)$ follows since $\sum\limits_{m = 0}^{M-1}\sum\limits_{{m' = 0, m'\neq m}}^{M-1} e^{jp(m-m')}$ in the second term can be replaced by
		$\sum \limits_{m = 0}^{M-1}\sum \limits_{m' = 0}^{M-1} e^{jp(m-m')}$ $-$ $\sum \limits_{m = 0}^{M-1} e^{0}$.
		From the derivation of \eqref{eq:Echin}, it holds that for $p \ne 0$ the second summand in \eqref{eq:abs2:Eabsx2_3} vanishes, resulting in
		\begin{equation}\label{eq:abs2:Eabsx2_4} 
		{\rm E}  \left[\sum \limits_{k = 0}^{K-1}\sum \limits_{k' = 0}^{K-1}e^{jpc_{n,k}-jpc_{n,k'}} \right] 
		=	\frac{(M-K)K}{M-1}. 
		\end{equation}
		Plugging (\ref{eq:abs2:Eabsx2_4}) into \eqref{eq:miss:abs2:start}, we obtain 
		\begin{equation}
		\label{eq:abs2:pn0}
		{\rm E}\left[  |\chi_n|^2\right] =
		\frac{M-K}{(M-1)K}u,\ \text{if\ } p \ne 0.
		\end{equation}
		\par
		
		Finally, to prove \eqref{eq:sumvecchi2}, we calculate $\sum_{n = 0}^{N-1}{\rm D}\left[\chi_n\right]$ for $p =0$ and $p \ne 0$.
		When $p = 0$, from \eqref{eq:abs2:p0}, we have that
		\begin{equation}
		\label{eq:sum:p0} 
		\sum_{n = 0}^{N-1}{\rm D}\left[\chi_n\right]
		=\sum_{n = 0}^{N-1} \left(u -u^2 \right) 
		=\left(u -u^2 \right)N. 
		\end{equation}
		When $p \ne 0$, it follows from \eqref{eq:abs2:pn0} that
		\begin{equation}
		\label{eq:sum:pn0}
		\sum_{n = 0}^{N-1}{\rm D}\left[\chi_n\right]
		=\frac{M-K}{(M-1)K}uN.
		\end{equation}
		Combining \eqref{eq:sum:p0} and \eqref{eq:sum:pn0} proves \eqref{eq:sumvecchi2}. 
		\qed

		\subsubsection{Proof of \eqref{eq:chibound2}}
		We again consider the two cases $p =0$ and $p \neq 0$ separately:
		When $p = 0$, it follows from \eqref{eq:Echin} that 
		\vspace{-0.1cm}
		\begin{align} 
		\left| \chi_n - {\rm E} \left[ \chi_n \right] \right|^2 
		&=  \left|\left( I_{n} - u \right) e^{jqn}\right|^2  =  \left( I_{n} - u \right)^2 \notag  \\
		&    =\begin{cases}
		(1-u)^2 \le 1,\ \text{if\ } I_n = 1,\\
		u^2 \le 1,\ \text{otherwise}.
		\label{eq:abschi:p0} 
		\end{cases} 
		\vspace{-0.1cm}
		\end{align}	
		When $p \ne 0$, $ \left| \chi_n - {\rm E} \left[ \chi_n \right] \right|^2  =   \left| \chi_n \right|^2$, which is not larger than 1 by definition of $\chi_n$ \eqref{eqn:Chindef}, thus proving \eqref{eq:chibound2}. 
		\qed

		\vspace{-0.2cm}
		\subsection{Proof of Theorem \ref{thm:MIP}}
		\label{app:Proof3}
		\vspace{-0.1cm}
		By fixing some positive $\epsilon \le V$, setting $t = \sqrt{\frac{N}{8V}}\epsilon$ and 
		\vspace{-0.1cm}
		\begin{equation}
		\label{eq:epsilonprime}
		\epsilon' := \frac{\sqrt{V}+\epsilon}{uN - \sqrt{\frac{N}{8V}}\epsilon } = \frac{\sqrt{V}+\epsilon}{uN - t },
		\vspace{-0.1cm}
		\end{equation}
		we have that for any $\left(\Delta_m, \Delta_n \right) \in \Xi$ 
		\vspace{-0.1cm}
		\begin{align} 
		\mathbb{P}\left(  \frac{|\chi|}{|\Lambda|} \ge \epsilon'  \right) 
		&\stackrel{(a)}{\leq} 	\mathbb{P}\left(\left| \chi \right| \ge \sqrt{V} + \epsilon \cup |\Lambda| \le uN - t  \right) \notag \\
		&\leq \mathbb{P}\left( \left| \chi \right| \ge \sqrt{V} \!+\! \epsilon \right) \!+\! \mathbb{P}\left( |\Lambda| \le uN \!- \!t \right)\notag\\
		&\stackrel{(b)}{\leq} 2e^{-\frac{\epsilon^2}{4V}}.
		\label{eq:barchibound} 
		\vspace{-0.1cm}
		\end{align}
		Here $(a)$ holds since the event $ \frac{|\chi|}{|\Lambda|} \ge \frac{\sqrt{V}+\epsilon}{uN - t }$ implies that at least one of the conditions $\left| \chi \right| \ge \sqrt{V} + \epsilon$ and $ |\Lambda| \le uN - t $ is satisfied; and $(b)$ follows from Corollary \ref{cor:Rayleigh} and Lemma \ref{lem:Lambdabound}.

		Using the bound \eqref{eq:barchibound} on the magnitude of the normalized correlation, we next bound the probability of  $\mu({\bm \Phi_{\ast}})$ to exceed some constant. By applying the union bound to \eqref{eqn:MIP2}, we have 
		\vspace{-0.1cm}
		\begin{align}
		\mathbb{P}\left( \mu({\bm \Phi_{\ast}}) \ge \epsilon'  \right)
		&\leq \sum \limits_{\left(\Delta_m, \Delta_n \right) \in \Xi} \mathbb{P}\left(\frac{|\chi\left(\Delta_m, \Delta_n \right)|}{|\Lambda|} >\epsilon' \right) \notag \\
		& \leq 2| \Xi| e^{-\frac{\epsilon^2}{4V}}. 
		\label{eq:pr_union} 
		\vspace{-0.1cm}
		\end{align}
		According to \eqref{eq:pr_union}, for any $\epsilon >0$, it holds that
		\vspace{-0.1cm}
		\begin{equation}
		\mathbb{P}\left(  \mu({\bm \Phi_{\ast}})  \leq \epsilon' \right) \ge 1- 2| \Xi| e^{-\frac{\epsilon^2}{4V}},
		\label{eq:tmp_Ne_bound}
		\vspace{-0.1cm}
		\end{equation}
		where $\epsilon'$ is obtained from $\epsilon$ via \eqref{eq:epsilonprime}. The right hand side of \eqref{eq:tmp_Ne_bound} is not smaller than $1-\delta$ when $\delta \ge 2| \Xi| e^{-\frac{\epsilon^2}{4V}}$, implying that $	\mathbb{P}\left(  \mu({\bm \Phi_{\ast}})  \leq \epsilon' \right)  \ge 1-\delta$ when $\epsilon$ satisfies 
		\vspace{-0.1cm}
		\begin{equation}
		\label{eq:epsilonbound}
		\epsilon \ge \sqrt{2V(\log 2|\Xi| - \log \delta)}. 
		\vspace{-0.1cm}
		\end{equation}	
		Finally, by \eqref{eq:epsilonprime}, fixing  $\epsilon'={1}/({2S-1})$ implies that
		\vspace{-0.1cm}
		\begin{equation} 
		S = \frac{uN - \sqrt{\frac{N}{8V}}\epsilon }{2\sqrt{V}+2\epsilon} + \frac{1}{2}
		= \frac{uN + \sqrt{N/8} }{2\sqrt{V}+2\epsilon} - \sqrt{ \frac{N}{32V}} + \frac{1}{2}.\label{eq:Sepsilon}
		\vspace{-0.1cm}
		\end{equation}
		Substituting  \eqref{eq:epsilonbound} into \eqref{eq:Sepsilon} proves \eqref{eq:Sbound}.
		\qed
		
	\end{appendix}
	\ifCLASSOPTIONcaptionsoff
	\newpage
	\fi

	
	
	%
	\bibliographystyle{IEEEtran}
	\bibliography{CAESAR_R1}

	
	

\end{document}